\documentclass[final,journal]{IEEEtran}

\usepackage{ifpdf}
\ifpdf
	\usepackage[pdftex]{graphicx}
  \usepackage[update]{epstopdf}
\else
	\usepackage{graphicx}
\fi
\usepackage{amsmath,amssymb,array,stfloats,cuted,midfloat,psfrag,setspace,mathbbol}
\usepackage[noadjust]{cite}
\usepackage{url}
\usepackage{hyperref}

\newtheorem{theorem}{\textbf{Theorem}}

\newtheorem{definition}{\textbf{Definition}}

\newtheorem{lemma}{\textbf{Lemma}}

\newcommand{\Expt}{\mbox{${\mathbb E}$} }

\usepackage{amsfonts}
\usepackage{amstext,mathbbol,mathrsfs}
\usepackage{latexsym}
\usepackage{color}
\usepackage{ifthen}
\usepackage{multirow}
\usepackage{subfigure}

\begin{document}
\date{}\title{Optimality and Approximate Optimality of Source-Channel Separation in Networks}

\author{Chao Tian,~\IEEEmembership{Senior Member,~IEEE}, Jun Chen,~\IEEEmembership{Member,~IEEE}, 
\\Suhas N. Diggavi,~\IEEEmembership{Fellow,~IEEE}, and Shlomo Shamai (Shitz),~\IEEEmembership{Fellow,~IEEE}
\thanks{The work of J. Chen was supported in part by an Early Researcher Award
from the Province of Ontario and in part by the Natural Science and
Engineering Research Council (NSERC) of Canada under a Discovery Grant.}
\thanks{The work of S. Diggavi was supported in part by NSF award 1136174 and MURI award AFOSR FA9550-09-064.}
\thanks{The work of S. Shamai was supported by the European
Commission in the framework of the FP7 Network of Excellence in
Wireless COMmunications NEWCOM++ and NEWCOM\#, and by the Israel Science Foundation
(ISF).} 
\thanks{This paper was presented in part at 2010 IEEE International Symposium on Information Theory, Austin, TX, Jun. 2010, and IEEE International Conference on Signal Processing and Communications, Bangalore, India, Jul. 2010.}
}
\maketitle

\vspace{2cm}

\begin{abstract}
We consider the source-channel separation architecture for lossy source coding in communication networks. 
It is shown that the separation approach is optimal in two general scenarios, and is approximately 
optimal in a third scenario. The two scenarios for which 
separation is optimal complement each other: the first is when the memoryless 
sources at source nodes are arbitrarily correlated, each of which is to be reconstructed at 
possibly multiple destinations within certain distortions, but 
the channels in this network are synchronized, orthogonal and memoryless point-to-point channels; the second
is when the memoryless sources are mutually independent, each of which is to be reconstructed only at 
one destination within a certain distortion, but the channels are general, including multi-user 
channels such as multiple access, broadcast, interference and relay channels, possibly with feedback. The third scenario, 
for which we demonstrate approximate optimality of source-channel separation, generalizes the second
scenario by allowing each source to be reconstructed at multiple destinations with different distortions. 
For this case, the loss from optimality by using the separation approach can be upper-bounded
when a ``difference" distortion measure is taken, 
and in the special case of quadratic distortion measure, this leads to universal constant bounds. 
\end{abstract}


\section{Introduction}
\label{sec:intro}

Shannon's source-channel separation theorem asserts that there is no
essential loss asymptotically in point-to-point communication systems, when the
source coding component and the channel coding component are designed and
operated separately \cite{Shannon:48}. This separation architecture
simplifies the overall communication system tremendously, because the
decoupled subsystems are much easier to design and implement, with the
codeword index as the only interface between the two components. Unfortunately, it has been
shown that the separation approach is not optimal in very simple
multiuser scenarios (e.g., \cite{Cover:80}), which suggests that the
optimality of source-channel separation may not hold beyond the
conventional point-to-point case.

Because of the clear benefits of the source-channel separation
architecture, it is important to understand the issue
better. In this work, we seek to answer the following sequence of
questions: is there a general class of multiuser communication systems
for which
\begin{itemize}
\item The separation approach is optimal?
\item If separation is not optimal, then is it at least approximately optimal?
\end{itemize}

The difficulty in answering these questions lies in the fact that in
most multiuser communication scenarios, we do not have explicit
characterizations of the rate-distortion regions, the channel capacity
regions, or the joint coding achievable distortion regions; however,
in order to determine whether the separation approach is optimal, it
is natural to first couple the rate-distortion region and the channel
capacity region, then compare it with the joint coding achievable
distortion region. With at least one region unknown in most cases, it
seems impossible to answer the above questions even in some of the
simplest settings (e.g., communicating sources on an interference
channel), let alone in more complex networks. In this work, we show that 
this difficulty in determining the optimality of source-channel separation can in fact be 
circumvented completely in several important settings, and the answers to 
the sequence of questions posed earlier are indeed positive.

More precisely, we show that for lossy coding of memoryless sources in a
network, the source-channel separation approach is optimal for the
following two general scenarios: the first scenario, referred
to as {\em distributed network joint source-channel coding} (DNJSCC), is when the sources are
arbitrarily correlated, each of which is to be reconstructed at
possibly multiple destinations within certain distortions, but the
channels between any pair of nodes  
in this network are synchronized, 
orthogonal, and memoryless; the second scenario, referred to as {\em joint source-channel multiple
unicast with distortions} (JSCMUD), is when the sources are mutually independent, each
of which is to be reconstructed only at one destination within a certain
distortion, but the channels can be general, including multi-user 
channels such as multiple access, broadcast, interference
and relay channels, possibly with feedback. 

The third scenario is a natural extension of the second one by allowing a source to 
be reconstructed at multiple destinations with different distortions; this case is referred to 
as {\em joint source-channel multiple multicast with distortions} (JSCMMD). 
For this scenario, the classical example of sending
a Gaussian source over a Gaussian broadcast channel \cite{Goblick65} reveals that 
 the source-channel separation approach is not optimal in general. Thus we turn our attention to 
whether the separation approach is approximately optimal, and show that under a ``difference" 
distortion measure, it is indeed so in the sense that the loss from the optimum can be upper-bounded. In the
important special case of quadratic distortion measure, the upper bound
is at most 0.5 bit per (additional) user which reconstructs the same source. 

The optimality of source-channel separation beyond point-to-point communications has been considered in the past for more restricted classes of sources and channels \cite{Yeung:95,XiaoLuo:07, Han:80, Han:10, SteinbergMerhav:06,Maor:06}, usually by taking advantage of the problem-specific structures and applying conventional techniques. The first scenario we consider, {\em i.e.}, the DNJSCC problem, is closely related to the problem treated in \cite{KoetterEffrosMedard:09}, where the optimality of the separation between channel coding and network coding \cite{Yeung:00} was established. In fact,  our interest in the DNJSCC problem was motivated by the success in this work,
from which we also borrow the ideas of channel simulation and sample interleaving; by applying these ideas directly, we obtain a concise proof for the DNJSCC problem without relying on the full-fledged stacked network as in \cite{KoetterEffrosMedard:09}, and our approach has the additional benefit of making explicit the underlying interactive source coding component. The result in \cite{KoetterEffrosMedard:09} was extended to the DNJSCC scenario in \cite{Jalali:10} independently from and concurrently with our work \cite{TianISIT:10,TianSPCOM:10,TianArxiv:10}. Another relevant work is \cite{AgarwalSahaiMitter:06} where the super-channel view similar to what we use  in the JSCMUD problem was applied to non-ergodic point-to-point channels. Also notable is the ``information separation'' discovered by Tuncel \cite{Tuncel:06}, which is a notion of separation weaker than the classical source-channel separation, and thus not the focus of this work. 

The rest of this paper is organized as follows. 
Examples are provided in Section \ref{sec:examples} to illustrate the underlying intuitions, and necessary definitions are
given in Section \ref{sec:def}. The main results and the proofs on DNJSCC, JSCMUD and JSCMMD 
are given in Sections \ref{sec:DNSC}, \ref{sec:proofunicast} and \ref{sec:proofmulticast},
respectively. Section \ref{sec:conclusion} finally concludes the paper.

\section{Three Examples}
\label{sec:examples}

In this section three examples are discussed in the context of sending
sources on interference channels to provide some intuitions for the
optimality or approximate optimality of source-channel separation in
DNJSCC, JSCMUD and JSCMMD. The main results of this work are built on these
intuitions, and Sections \ref{sec:DNSC}, \ref{sec:proofunicast} and
\ref{sec:proofmulticast} essentially make them more precise and
rigorous. For simplicity, the channel
bandwidth and the source bandwidth are assumed to match in this section.

\subsection{An Example for Distributed Network Joint Source-Channel Coding}

\begin{figure}[tb]
\begin{centering}
\includegraphics[width=8cm]{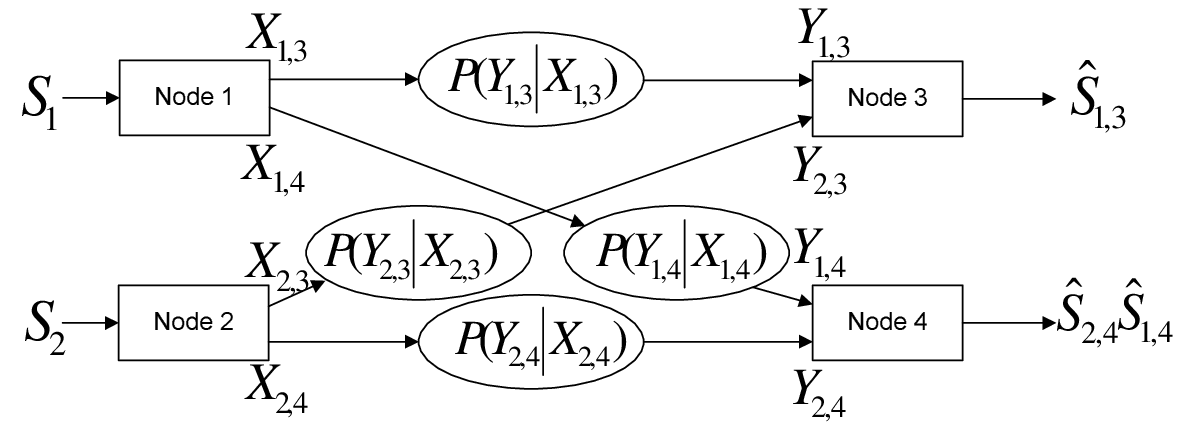}
\caption{Transmitting correlated sources on an interference
  network.\label{fig:example0}}
\end{centering}
\end{figure}

Consider the example in Fig. \ref{fig:example0}, where the
discrete-time finite-alphabet memoryless sources $S_1$ and $S_2$ are correlated. Each
discrete-time finite-alphabet memoryless channel between a transmitter and a receiver is
orthogonal to the other channels: the channel from node $i$ to node $j$ has 
transition probability $P(Y_{i,j}|X_{i,j})$ and channel capacity $C_{i,j}$, and 
the overall transition probability of the channel network is $\prod_{(i,j)} P(Y_{i,j}|X_{i,j})$. 
Both node $3$ and node $4$ require a lossy reconstruction of
source $S_1$, denoted as $\hat{S}_{1,3}$ and $\hat{S}_{1,4}$, respectively. Node $4$ also requires
a lossy reconstruction of source $S_2$, denoted as $\hat{S}_{2,4}$. The rate-distortion region of 
the underlying source coding problem is unknown, characterizing which is at least as difficult as the distributed
source coding problem \cite{Berger:78}.

\begin{figure}[tb]
\begin{centering}
\includegraphics[width=8cm]{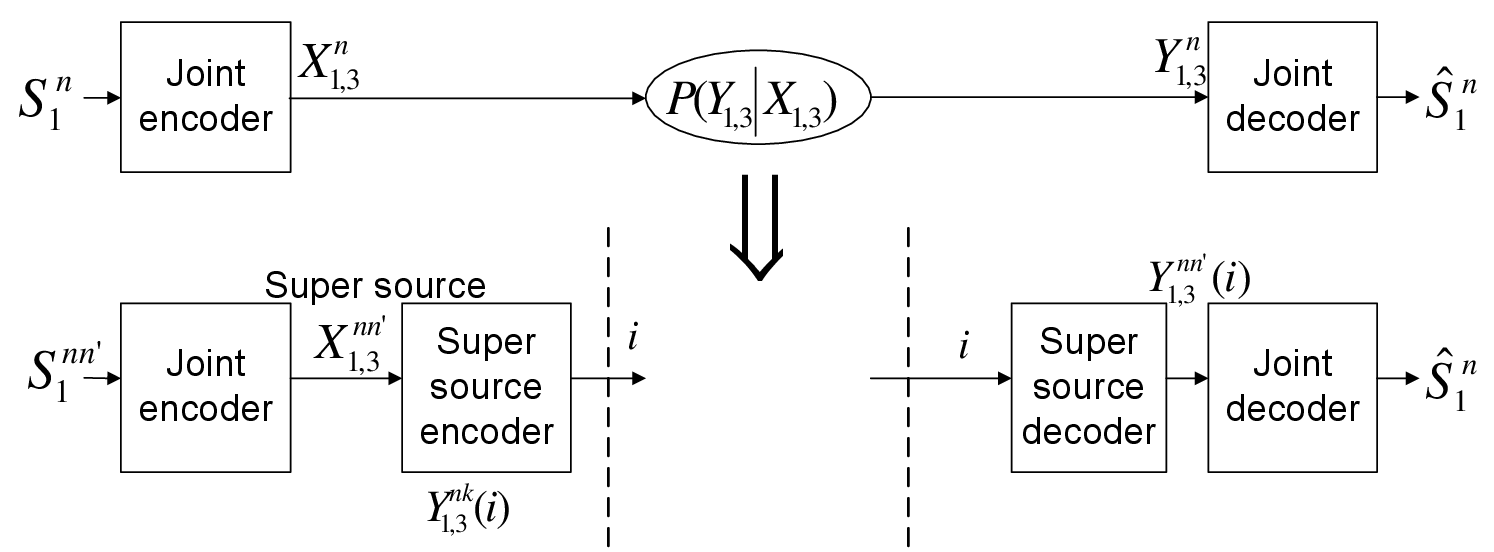}
\caption{Extracting a super-source from a joint source-channel code.\label{fig:DNStrans}}
\end{centering}
\end{figure}

Suppose there exists a length-$n$ joint
source-channel code that achieves the distortion triple
$(D_{1,3},D_{1,4},D_{2,4})$. The key observation is the following
simple fact: if we fix this joint source-channel code, 
then the channel input for any given channel, for example $X^n_{1,3}$, 
can be viewed as a super (block) source, independent and identically distributed 
across blocks; see Fig. \ref{fig:DNStrans}. Therefore, we can encode a length-$n'$ sequence of such blocks
 using a ``rate-distortion'' code of rate per block
slightly exceeding $I(X^n_{1,3};Y^n_{1,3})$, the codewords of which are generated using
the distribution $P(Y^n_{1,3})$. It follows that with probability approaching one (as $n'$ goes to infinity)  a 
$Y^{nn'}_{1,3}$ codeword can be found in the codebook that is jointly typical with a channel input sequence $X^{nn'}_{1,3}$, 
{\em i.e.}, a length-$n'$ vector of the super source samples. 
This lossy source code essentially simulates the channel output over $n'$ length-$n$ blocks, and only the codeword index needs 
to be known at node $3$ to reconstruct the simulated channel output $Y^{nn'}_{1,3}$. 
Note that the rate of this code is $I(X^n_{1,3};Y^n_{1,3})\leq nC_{1,3}$; a similar argument holds for all other links. The original joint source-channel code decoders can now be applied on the simulated channel outputs to yield the reconstructions. This intuitively implies 
that the underlying {\em source coding problem} is guaranteed to achieve the distortion $(D_{1,3},D_{1,4},D_{2,4})$ at rates $(C_{1,3},C_{1,4},C_{2,3},C_{2,4})$, 
which would further imply the optimality of the separation approach.    

The above observation largely reflects the intuition behind the proof of source-channel separation for the DNJSCC
problem, however, some technical
details (besides the asymptotically diminishing quantities omitted in
the above discussion) need to be addressed: the main difficulty is
that when the network has relays or cycles, the super source argument
given above does not apply since channel usage constraints prevent 
coding over long super-channel blocks directly. The proof
given in Section \ref{sec:DNSC} will resolve this difficulty through
an intricate arrangement of channel simulation.

\subsection{An Example for Joint Source-Channel Multiple Unicast with Distortions}

\label{subsec:MUMMD}

Consider the problem depicted in Fig. \ref{fig:example1}, where the sources $S_1$, $S_2$ and $S_3$ are mutually independent; here the interference channel is more generally
given by the transition probability $P(Y_3,Y_4|X_1,X_2)$, where $X_1,X_2$ are the channel inputs by node $1$ and node $2$, respectively, and $Y_3,Y_4$ are the channel outputs at
node $3$ and node $4$, respectively. Since the capacity region of the interference channel is unknown, it is infeasible
to explicitly characterize the achievable distortion region of the separation approach.

\begin{figure}[tb]
\begin{centering}
\includegraphics[width=8cm]{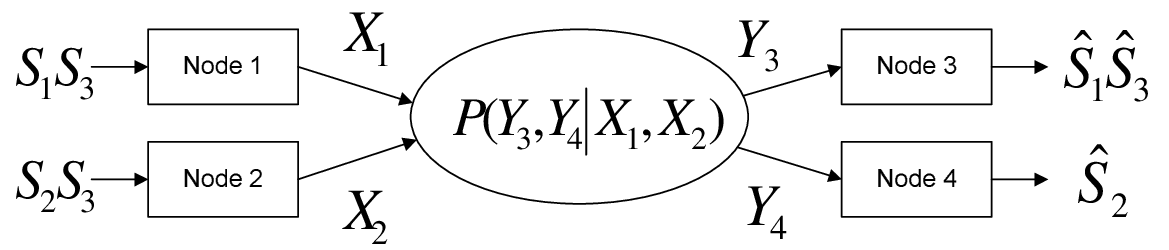}
\caption{Transmitting mutually independent $S_1,S_2,S_3$ on an interference channel.\label{fig:example1}}
\end{centering}
\end{figure}

Suppose a distortion triple $(D_1,D_2,D_3)$ is achievable using some joint source-channel
code of length-$n$. 
The key observation is now the following
simple fact: if we fix this joint
source-channel code, the transition probability of
$P(\hat{S}^n_1,\hat{S}^n_2,\hat{S}^n_3|S^n_1,S^n_2,S^n_3)$ can be viewed as that of
an alternative super interference channel with three users. On this super channel, 
the individual mutual information guarantee 
$I(S^n_i;\hat{S}^n_i)\geq n R_{i}(D_i)$ holds for $i=1,2,3$,
due to the conventional rate-distortion theorem \cite{CoverThomas}. Thus intuitively, this super channel is \lq\lq{}good\rq\rq{} 
since the mutual information $I(S^n_i;\hat{S}^n_i)$ terms are lower bounded, and the rate triple
$(nR_{1}(D_1),nR_{2}(D_2),nR_{3}(D_3))$ should be in its capacity region, 
which would further imply that any achievable distortion triple $(D_1,D_2,D_3)$ is 
achievable by the separation approach.

In order to show that the super interference channel can indeed 
support the rate triple $(nR_{1}(D_1),nR_{2}(D_2),nR_{3}(D_3))$, we
essentially need to construct (random) codes over large super-channel 
blocks, and prove that the error 
probability can be made small, just as in conventional channels. The
proof in Section \ref{sec:proofunicast} follows this approach and
makes the above intuitive argument more rigorous.

\subsection{An Example for Joint Source-Channel Multiple Multicast with Distortions}

\begin{figure}[tb]
\begin{centering}
\includegraphics[width=8cm]{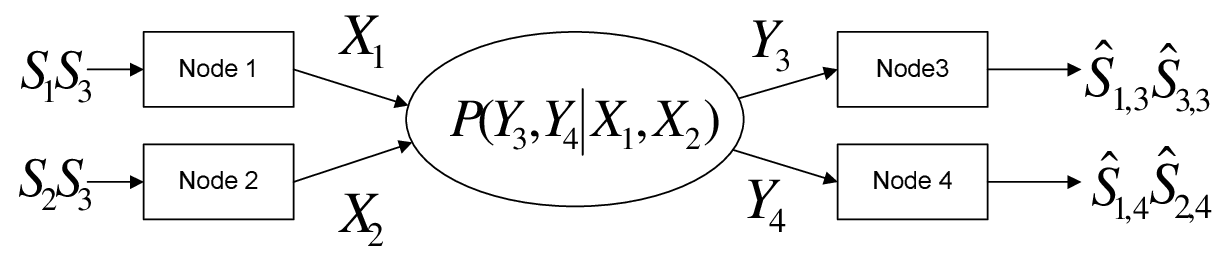}
\caption{Transmitting mutually independent $S_1,S_2,S_3$ on an interference channel
  to multiple destinations, {\em i.e.}, source $S_1$ is required at both
  destination node $3$ and node $4$.\label{fig:example2}}
\end{centering}
\end{figure}

Consider the problem depicted in Fig. \ref{fig:example2}, which is only slightly different from
that in Fig. \ref{fig:example1} in that source $S_1$ is to be reconstructed at both node $3$ and node $4$, denoted as $\hat{S}_{1,3}$ and $\hat{S}_{1,4}$, respectively; the reconstruction of source $S_3$ at node $3$ is denoted\footnote{The notation used here may seem unnatural initially, however it will become clear that this notation is convenient when generalizing to more complex networks.} as $\hat{S}_{3,3}$ and the reconstruction of source $S_2$ at node $4$ is denoted as $\hat{S}_{2,4}$. 
Taking a similar view as in the previous example, the abstracted channel now has transition probability 
$P(\hat{S}^n_{1,3},\hat{S}^n_{1,4},\hat{S}^n_{2,4},\hat{S}^n_{3,3}|S^n_1,S^n_2,S^n_3)$. 
However, the mutual information bounds by the conventional
rate-distortion theorem 
cannot be directly used as in the previous case. A
moment of thought should convince the readers that the broadcast
nature of the marginal transition probability
$P(\hat{S}^n_{1,3},\hat{S}^n_{1,4}|S^n_1)$ is the culprit, and some
additional coding component is needed. 

A natural separation architecture here 
is to use a successive refinement source code  \cite{EquitzCover:91} to produce
descriptions satisfying the distortion requirements for each destination and couple it to a superposition
broadcast code \cite{CoverThomas} to deliver reliably these 
messages in the degraded message set \cite{KornerMarton:77}. More precisely, in the example of Fig.
\ref{fig:example2}, assume without loss of generality that 
the distortion for source $S_1$ at node $3$ is greater than 
that at node $4$. A successive refinement code for
$S_1$ can be used to produce messages $(W_{1,1},W_{1,2})$ such that $W_{1,1}$ is
to be delivered to node $3$ and both $(W_{1,1},W_{1,2})$ are to be
delivered to node $4$. Node $1$ also produces a message $W_{3,1}$ 
to encode source $S_3$, and node $2$ produces a message $W_{2,1}$ to encode
source $S_2$. The messages $(W_{1,1},W_{3,1})$ need to be reliably transmitted  
to node $3$, and the messages $(W_{1,1},W_{1,2},W_{2,1})$ to node $4$. 

Let us for the moment isolate source $S_1$ and focus on the super block broadcast channel 
$P(\hat{S}^n_{1,3},\hat{S}^n_{1,4}|S^n_1)$ with the messages $(W_{1,1},W_{1,2})$, since it is the main difficulty in generalizing the proof approach for JSCMUD. 
We can show that this broadcast channel can support a certain rate pair for degraded message set broadcast, by introducing an additional auxiliary random variable. The same auxiliary random variable is also used to construct 
successive refinement source code for $S_1$. The afore-mentioned broadcast channel code rates are however insufficient to support this 
successive refinement source code; 
nevertheless, the shortfall can be upper-bounded 
by comparing the channel code rates and the source code rates. This upper bound implies the 
approximate optimality of source-channel separation in JSCMMD.

\section{Notation and Definitions}
\label{sec:def}

In this section, notation and necessary definitions are provided. 
The notation would become rather unwieldy if a unified framework were 
used for all the problems treated in this
work, therefore we forgo this ambitious goal and define the problems
separately. We focus on the problems with discrete-time finite-alphabet memoryless sources, 
 discrete-time finite-alphabet memoryless channels and bounded 
distortion measures, unless stated otherwise explicitly. 
It should be noted that it is often assumed
that the sources are independent of the channels in such separation problems, 
which is also assumed in this work;
this is because otherwise, even if the encoding and the decoding
functions are designed separately, the inherent dependence between the source and the channel will render such a
separation rather meaningless even in a point-to-point setting.

\subsection{Definitions for the Distributed Network Joint Source-Channel Coding Problem}

For this case, the network with a total of $N$ nodes can be
conveniently written as a directed graph
$\mathcal{G}=(\mathcal{V},\mathcal{E})$, where
$\mathcal{V}=\{1,2,\ldots,N\}$ is the set of
nodes, and $\mathcal{E}$ is the set of edges between any two nodes; from here on, 
for any positive integer $M$, we use $\mathcal{I}_M$ to denote the set $\{1,2,\ldots,M\}$. 

Each edge $e=(i,j)\in \mathcal{E}$ is associated with a channel, 
whose transition probability is given as $P(Y_{i,j}|X_{i,j})$
with input alphabet $\mathcal{X}_{i,j}$ and output alphabet
$\mathcal{Y}_{i,j}$ where the input and the output are not always independent, 
{\em i.e.}, the capacity of the channel on this link is non-zero; 
these channels are assumed to be synchronized. Each node $i$ has a source $S_i$,
distributed in the alphabet $\mathcal{S}_i$, and the collection of the
sources are distributed according to the joint distribution
$P(S_1,S_2,\ldots,S_N)$ at each time instance. We have
inherently assumed these sources are synchronized for simplicity, and thus the
notation $P(S_1,S_2,\ldots,S_N)$ is meaningful. A length-$n$ vector of a source $S_i$ is written as
$S^n_i$, and the $t$-th symbol in this vector is written as $S_i(t)$; {\em i.e.}, $S^n_{i}=(S_i(1),S_i(2),\ldots,S_i(n))$. A set of sources 
$\{S_{i},i\in \mathcal{A}\}$ may be written as $S_{\mathcal{A}}$; similarly, $\{X_{i,j},(i,j)\in \mathcal{A}\}$ may be written as $X_{\mathcal{A}}$. 
Upper case is used for random variables, and lower case for
their realizations. For any set $\mathcal{S}$, its $r$-th order 
product set is written as $\mathcal{S}^r$. 

For each source, a distortion measure is defined 
as $d: \mathcal{S}_i\times \hat{\mathcal{S}}_i\rightarrow [0,\infty)$
 where $\hat{\mathcal{S}}_i$ is the reconstruction alphabet. Nodes that are interested in a given source $S_i$ may
use different reconstruction alphabets and distortion measures, however, we do not distinguish 
them for notational simplicity.  A node $j$ may be interested in only a subset of the sources
$\{S_i,i\in \mathcal{I}_N\}$; notationally, the set of
sources that node $j$ is interested in is written as $\mathscr{T}_j$. The class of codes being considered for the distributed network
source coding problem are conventional block codes defined below.

\begin{definition}
\label{def:DNSj1}
An $(m,n,\{d_{k,j},k\in \mathscr{T}_j\})$ distributed network joint source-channel 
code on a joint source-channel network $(\mathcal{V},\mathcal{E},\{\mathscr{T}_j,j\in\mathcal{I}_N\},P(S_1,S_2,\ldots,S_N),\prod_{(i,j)\in \mathcal{E}}P(Y_{i,j}|X_{i,j}))$ consists of the following components:
\begin{itemize}
\item At each transmitter node $i$, for each $j$ such that $(i,j)\in
  \mathcal{E}$, an encoding function for time instance $t$
\begin{align}
\label{eq:JSCencDSC}
&\phi_{i,j}^{(t)}: \mathcal{S}^{m}_i\times \prod_{(k,i)\in
    \mathcal{E}} \mathcal{Y}^{t-1}_{k,i}\rightarrow \mathcal{X}_{i,j},
  \qquad t =1,2,\ldots, n.
\end{align}
\item At each receiver node $j$, for each source $k\in \mathscr{T}_j$,
  a decoding function
\begin{align}
\label{eqn:sourcedecoderDSC}
\psi_{k,j}: \prod_{(i,j)\in \mathcal{E}}\mathcal{Y}^{n}_{i,j}\times
\mathcal{S}^{m}_j \rightarrow \hat{\mathcal{S}}^{m}_{k}.
\end{align}
\end{itemize}
The encoding and the decoding functions induce the distortions
\begin{align*}
&d_{k,j}=\frac{1}{m}\sum_{t=1}^{m}\Expt{d(S_{k}(t),\hat{S}_{k,j}(t))},\\
&\qquad\qquad j=1,2,\ldots,N,\quad\mbox{and} \quad k\in \mathscr{T}_j,
\end{align*}
where $\hat{S}_{k,j}$ is the reconstruction of source $S_k$ at node $j$.
\end{definition}

Here $m$ is the source block length and $n$ is the channel block length, which imply that 
there is a source-channel bandwidth mismatch factor of $\kappa=n/m$ (channel uses per source sample).
If a node is not interested in a certain source, the distortion of the reconstruction at this node can simply be assumed to be
large. Thus we can write a distortion matrix, whose element $d_{k,j}$ is the distortion associated with
the reconstruction of source $S_k$ at node $j$. Without loss
of generality\footnote{Without loss of generality, we can always
  assume the minimum distortion for a given distortion measure is
  zero; see \cite{BergerBook}.}, let the element $d_{i,i}=0$
and define $d_{i,j}=d^{\max}_i$ for $i\notin\mathscr{T}_j$,
where $d^{\max}_i$ is the distortion achievable at rate zero for
source $S_i$. The region of achievable distortion matrices can be defined as follows.

\begin{definition}
\label{def:DNSj2}
A distortion matrix $\vec{D}$ is achievable for distributed network joint 
source-channel coding with bandwidth mismatch factor $\kappa$ on a joint 
source-channel network $(\mathcal{V},\mathcal{E},\{\mathscr{T}_j,j\in\mathcal{I}_N\},P(S_1,S_2,\ldots,S_N),\prod_{(i,j)\in \mathcal{E}}P(Y_{i,j}|X_{i,j}))$, 
if for any
$\epsilon>0$ and sufficiently large $m$, there exist an integer $n\leq
\kappa m $ and an $(m,n,\{d_{k,j},k\in \mathscr{T}_j\})$ distributed 
network joint source-channel code, such that $d_{i,j}\leq D_{i,j}+\epsilon$,
$i,j=1,2,\ldots,N$. The collection of all such distortion matrices is
the distributed network joint source-channel coding achievable distortion region,
denoted as $\mathcal{D}_{dis}$.
\end{definition}

To discuss source-channel separation, it is important to define the source coding problem and the 
channel coding problem that are being separated into. The channel coding problem in DNJSCC is simply the point-to-point channel
capacity problem. The source coding problem is more complex, which requires  the incorporation of 
 interactive coding.

\begin{definition}
\label{def:DDNSC}
An $(m,l,\{L_{i,j},(i,j)\in\mathcal{E}\},\{d_{k,j},k\in \mathscr{T}_j\})$ distributed network source code with a total of $l$ sessions on a source communication
network $(\mathcal{V},\mathcal{E},\{\mathscr{T}_j,j\in\mathcal{I}_N\},P(S_1,S_2,\ldots,S_N))$ consists of the following components:
\begin{itemize}
\item At each (transmitter) node $i$, for each $j$ such that $(i,j)\in \mathcal{E}$, an encoding function for transmission session $t=1,2,\ldots,l$,
\begin{align}
&\tilde{\phi}_{i,j}^{(t)}: \mathcal{S}^{m}_i\times \prod_{(k,i)\in \mathcal{E}} \mathcal{I}_{L_{k,i}}^{t-1}\rightarrow \mathcal{I}_{L_{i,j}}, 
\end{align}
where $L_{i,j}$ and $L_{k,i}$'s are positive integers.
\item At each receiver node $j$, for each source $k\in \mathscr{T}_j$, a decoding function
\begin{align}
\label{eqn:sourcedecoderDSCdigital}
\tilde{\psi}_{k,j}: \prod_{(i,j)\in \mathcal{E}}\mathcal{I}^l_{L_{i,j}}\times \mathcal{S}^{m}_j \rightarrow \hat{\mathcal{S}}^{m}_{k}.
\end{align}
\end{itemize}
The encoding functions and the decoding functions induce the distortions
\begin{align*}
&d_{k,j}=\frac{1}{m}\sum_{t=1}^{m}\Expt{d(S_{k}(t),\hat{S}_{k,j}(t))},\nonumber\\
&\qquad\qquad j=1,2,\ldots,N,\quad\mbox{and} \quad k\in \mathscr{T}_j,
\end{align*}
where again $\hat{S}_{k,j}$ is the reconstruction of source $S_k$ at node $j$.
\end{definition}

\begin{definition}
\label{def:sourcecodingRD}
A rate-distortion-matrix tuple $(\{R_{i,j},
(i,j)\in\mathcal{E}\},\vec{D})$ is achievable on a source communication
network $(\mathcal{V},\mathcal{E},\{\mathscr{T}_j,j\in\mathcal{I}_N\},P(S_1,S_2,\ldots,S_N))$, if for any $\epsilon>0$, there exists an
integer $l$, such that for any sufficiently large $m$, there exists an
$(m,l,\{L_{i,j},(i,j)\in\mathcal{E}\},\{d_{k,j},k\in \mathscr{T}_j\})$
distributed network source code such that
\begin{align}
&R_{i,j}+\epsilon\geq \frac{l}{m}\log L_{i,j},\quad (i,j)\in\mathcal{E}\nonumber\\
&d_{i,j}\leq D_{i,j}+\epsilon, \quad i,j=1,2,\ldots,N.
\end{align}
The collection of distortion matrices $\vec{D}$ for which the
rate-distortion-matrix tuple $(\{R_{i,j},
(i,j)\in\mathcal{E}\},\vec{D})$ is achievable for a given rate vector
$\{R_{i,j}, (i,j)\in\mathcal{E}\}$ is denoted\footnote{$\mathcal{D}_{dis}$ has already been used in the joint coding problem, and here we slightly abuse the notation by using $\mathcal{D}_{dis}(\{R_{i,j}\}_{(i,j)\in\mathcal{E}})$ to denote the distortion-rate function in the source coding problem. } as
$\mathcal{D}_{dis}(\{R_{i,j}\}_{(i,j)\in\mathcal{E}})$.
\end{definition}

Note that in the above definition, $m$ grows to infinity for any fixed value of $l$. One may alternatively define 
the region to allow $m$ and $l$ to grow in a more general manner. However, this alternative definition 
will only enlarge the region $\mathcal{D}_{dis}(\{R_{i,j}\}_{(i,j)\in\mathcal{E}})$, and thus does not affect the optimality result. 
In other words, the separation result we shall present is in fact stronger with the restrictions in Definitions \ref{def:DDNSC} and \ref{def:sourcecodingRD} 
than that under a more general version of these definitions.

Roughly speaking, $\frac{1}{m}\log L_{i,j}$ is the rate of the noiseless channel on edge $(i,j)$ 
per source symbol in each session. 
There are a total of $l$ sessions, and on each edge the same rate is used in all sessions. 
At the end of each session, the index $w_{j,k}\in \mathcal{I}_{L_{j,k}}$ 
in this session becomes available at destination node $k$, which can be used by node $k$ 
 in the next session. 
In other words, the encoding functions observe the causality
constraints on the session level.
Note that the region $\mathcal{D}_{dis}(\{R_{i,j}\}_{(i,j)\in\mathcal{E}})$ is convex by a time-sharing argument. Definitions \ref{def:DDNSC} and \ref{def:sourcecodingRD} 
specify a special class of interactive source coding problem, which appears 
particularly important given the result presented in this work.  

We can now combine the 
source codes together with the capacity-achieving channel codes for
each channel on the original communication network. More precisely, we
can define the achievable distortion region using such a separation
approach as
\begin{align}
\label{eqn:DisRegionDef}
\mathcal{D}^*_{dis}=\mathcal{D}_{dis}(\{\kappa C_{i,j}\}_{(i,j)\in\mathcal{E}}),
\end{align}
where $C_{i,j}$ is the channel capacity between node $i$ and node $j$, sometimes written as $C_e$ with $e=(i,j)\in \mathcal{E}$.

\subsection{Definitions for Joint Source-Channel Multiple Unicast and Multiple Multicast with Distortions}
\label{subsec:DefsUnicastMulticast}

There are $M$ mutually independent sources, denoted as $S_i$, distributed in the alphabet
$\mathcal{S}_i$ according to some distribution $P(S_i)$,
$i=1,2,\ldots,M$; note that the index $i$ here is not related to the
index of the node, unlike in the last section. For simplicity, we assume all the sources
are synchronized. The distortion measures are defined similarly as in the
last subsection, however we do not allow the existence of multiple
distortion measures for the same source. 
Let the number of nodes be $N$. 
For simplicity we treat the overall communication network as a single memoryless channel, with inputs 
$(X_1,X_2,\ldots,X_N)$ over the alphabets $\mathcal{X}_1\times\mathcal{X}_2\times\ldots\times\mathcal{X}_N$ and outputs $(Y_1,Y_2,\ldots,Y_N)$ over the alphabets  $\mathcal{Y}_1\times\mathcal{Y}_2\times\ldots\times\mathcal{Y}_N$, and transition probability given by $P(Y^N_1|X^N_1)$; $X_i$ and $Y_i$ are the channel input and output at node $i$, respectively. 

Each source $S_i$ can be present at several nodes, and for
each node $j\in \mathcal{I}_N$, we denote the sources present at node
$j$ as $\mathscr{S}_j$. The receiver demands are defined as follows:
\begin{itemize}
\item \textbf{Joint source-channel multiple unicast with distortions:} each source is to be reconstructed
 at a single destination. Again denote
  for receiver node $j$ the set of the sources it is interested in as
  $\mathscr{T}_j$, then $\mathscr{T}_j\cap
  \mathscr{T}_k=\emptyset$ for any $j\neq k$.
\item \textbf{Joint source-channel multiple multicast with distortions:} each source is to be reconstructed
 at multiple destinations, {\em i.e.}, it is
  possible that $\mathscr{T}_j\cap \mathscr{T}_k \neq\emptyset$.
\end{itemize}

\begin{definition}
\label{def:JScodeUnicast}
An $(m,n,d_1,d_2,\ldots,d_M)$ JSCMUD code on a source-channel 
 communication network $(\{\mathscr{S}_j,j\in\mathcal{I}_N\},\{\mathscr{T}_j,j\in\mathcal{I}_N\},\prod_{i=1}^M P(S_i),P(Y^N_1|X^N_1))$ consists of the following components:
\begin{itemize}
\item At each transmitter node $j$, an encoding function for (time) index $t$
\begin{align}
\label{eq:JSCencUnicast}
&\phi_{j}^{(t)}:\prod_{i\in \mathscr{S}_j} \mathcal{S}^{m}_i\times 
  \mathcal{Y}^{t-1}_{j}\rightarrow \mathcal{X}_{j},
  \qquad t =1,2,\ldots,n.
\end{align}
\item At each receiver node $j$, for each source $k\in \mathscr{T}_j$,
  a decoding function
\begin{align}
\label{eqn:sourcedecoder}
\psi_{k,j}: \mathcal{Y}^{n}_{j}\times
\prod_{i\in \mathscr{S}_j}\mathcal{S}^{m}_i \rightarrow \hat{\mathcal{S}}^{m}_k.
\end{align}
\end{itemize}
The encoding functions and decoding functions induce the distortion
\begin{align*}
d_k=\frac{1}{m}\sum_{t=1}^{m}\Expt{d(S_{k}(t),\hat{S}_k(t))},\quad k=1,2,\ldots,M,
\end{align*}
where $\hat{S}_k(t)$ is the reconstruction of source $S_k$ at a node $j$ such that $k\in\mathscr{T}_j$.
\end{definition}

\begin{definition}
\label{def:JScodeUnicastRegion}
A distortion vector $(D_1,D_2,\ldots,D_M)$ is achievable for JSCMUD  on a source-channel 
 communication network $(\{\mathscr{S}_j,j\in\mathcal{I}_N\},\{\mathscr{T}_j,j\in\mathcal{I}_N\},\prod_{i=1}^M P(S_i),P(Y^N_1|X^N_1))$ with a bandwidth mismatch factor $\kappa$, if for any $\epsilon>0$ and sufficiently large $m$, there exist an integer $n\leq \kappa m$ and an $(m,n,d_1,d_2,\ldots,d_M)$ JSCMUD code, such that $d_i\leq D_i+\epsilon$,
$i=1,2,\ldots,M$. The collection of all such distortion vectors is the achievable JSCMUD distortion region, denoted as $\mathcal{D}_{uni}$.
\end{definition}

Next we define the source coding problem and the channel coding problem that are being separated into. For the JSCMUD problem, 
the source codes are conventional lossy source
codes. The channel coding problem is more involved: each source $S_i$ is replaced with a message $W_i$ of
cardinality $L_i$ with a uniform distribution; moreover, these
messages are mutually independent. The precise channel code definition is as follows.

\begin{definition}
\label{def:DigChannelcodeUnicast}
An $(n,L_1,L_2,\ldots,L_M,P_{err})$ multiple unicast channel code on a channel communication network 
 $(\{\mathscr{S}_j,j\in\mathcal{I}_N\},\{\mathscr{T}_j,j\in\mathcal{I}_N\},P(Y^N_1|X^N_1))$
consists of the following components:
\begin{itemize}
\item At each transmitter node $j$, an encoding function for (time) index $t$
\begin{align}
\label{eq:DigEncUnicast}
&\tilde{\phi}_{j}^{(t)}:\prod_{i\in \mathscr{S}_j} \mathcal{I}_{L_i}\times
  \mathcal{Y}^{t-1}_{j}\rightarrow
  \mathcal{X}_{j},\qquad t =1,2,\ldots, n.
\end{align}
\item At each receiver node $j$, for each message $W_k$ where $k\in \mathscr{T}_j$,
  a decoding function
\begin{align}
\label{eq:DigSuperDecUnicast}
\tilde{\psi}_{k,j}: \mathcal{Y}^{n}_{j}\times\prod_{i\in
  \mathscr{S}_j} \mathcal{I}_{L_i} \rightarrow \mathcal{I}_{L_k}.
\end{align}
\end{itemize}
Denote the decoded message as $\hat{W}_i$ at node $j$ where $i\in \mathcal{T}_j$. The encoding functions and decoding functions induce the average decoding error probability
\begin{align}
P_{err}=\mbox{Pr}(\bigcup_{i=1}^M{W_i\neq \hat{W}_i}).
\end{align}
\end{definition}

\begin{definition}
A rate vector $(R_1,R_2,\ldots,R_M)$ is achievable for multiple
unicast channel coding on a channel communication network 
 $(\{\mathscr{S}_j,j\in\mathcal{I}_N\},\{\mathscr{T}_j,j\in\mathcal{I}_N\},P(Y^N_1|X^N_1))$, if for any
$\epsilon>0$ and sufficiently large $n$, there exists an
$(n,L_1,L_2,\ldots,L_M,\epsilon)$ multiple unicast channel code, such that $R_i\leq \frac{1}{n}\log
L_i+\epsilon$, $i=1,2,\ldots,M$. The collection of such achievable rate vectors is the
achievable capacity region of the network, denoted as $\mathcal{C}_{uni}$.
\label{def:diguniregion}
\end{definition}

Using conventional rate-distortion codes on each source and then
combining it with the above defined multiple unicast channel codes, an achievable distortion region is immediate, which
will be denoted as $\mathcal{D}^*_{uni}$. More precisely, we can write
\begin{align}
&\mathcal{D}^*_{uni}=\bigcup_{(R_1,R_2,\ldots,R_M)\in\mathcal{C}_{uni}}\left\{\begin{array}{c}(D_1,D_2,\ldots,D_M):\\
D_i\geq D_i(\kappa R_i),\\
i=1,2,\ldots,M\end{array}\right\}
\label{eqn:Duni}
\end{align}
where $D_i(\cdot)$ is the distortion-rate function of the source $S_i$.

In the case of JSCMMD, a source is to be
reconstructed with possibly different distortions at multiple
destinations. The JSCMMD codes are defined in the
same manner as in the case of JSCMUD, and thus the
detailed definitions are omitted here. The
achievable distortion matrix and the achievable distortion
region $\mathcal{D}_{mul}$ can also be defined accordingly.

The source-channel separation scheme for JSCMMD is slightly more
involved. Consider first source $S_i$, and assume it is to be
reconstructed in a lossy manner at nodes in the set
$\mathscr{Q}_i=\{j:i\in\mathscr{T}_j\}$. The source codes we shall
consider are successive refinement codes \cite{EquitzCover:91}, and
source $S_i$ is encoded in $|\mathscr{Q}_i|$ stages, where the
operator $|\cdot|$ denotes the cardinality of a set. For the channel 
codes in the separation approach, 
we consider the degraded message set problem \cite{KornerMarton:77}. More
precisely, in the given communication network, fix an  order $O_i$ for the elements in the set $\mathscr{Q}_i$ for each $i=1,2,\ldots,M$.
The source $S_i$ is replaced with a total of $|\mathscr{Q}_i|$ messages, denoted
as $W_{i,j}$, whose rate is $R_{i,O_i(j)}$, $j=1,2,\ldots,|\mathscr{Q}_i|$, where $O_i(j)$ is the $j$-th 
element in the order $O_i$. The
$k$-th node in this given order $O_i$ is required to reconstruct the first
$k$ messages, $W_{i,j}$, $j=1,2,\ldots,k$. We can now define 
the achievable capacity region $\mathcal{C}_{mul}(O_1,O_2,\ldots,O_M)$ for this
degraded message set problem, which depends on the set of
orders $\vec{O}=(O_1,O_2,\ldots,O_M)$; see the JSCMMD example 
in Section \ref{subsec:MUMMD}, where  $\mathscr{Q}_1=\{3,4\}$ and the specific order 
discussed is $O_1=(3,4)$.

The degraded message set problem naturally sets the stage
for the successive refinement source codes, and by combining these two
components, we arrive at an achievable distortion region using the 
 separation appraoch for a given 
set of orders $\vec{O}$. We
shall denote this achievable region as
$\mathcal{D}^*_{mul}(\vec{O})$.

\section{Optimality of Separation for Distributed Network Joint Source-Channel Coding}
\label{sec:DNSC}

Our first main result formally states the optimality of source-channel separation in the DNJSCC problem. Recall $\mathcal{D}_{dis}$ and $\mathcal{D}^*_{dis}$ given in Definition \ref{def:DNSj2} and Eqn. (\ref{eqn:DisRegionDef}), respectively.
\begin{theorem}
\label{theorem:DNSC}
$\mathcal{D}_{dis}=\mathcal{D}^*_{dis}$.
\end{theorem}

The uniform Markov lemma \cite{Chang:78,Kaspi:79} is needed in the proof of this theorem, which is an alternative version of the Markov lemma in \cite{Berger:78,Tung:78}. 
It is rewritten below using notation more convenient to us.
\begin{lemma}
\label{lemma:markov}
Let $X\leftrightarrow Y \leftrightarrow Z$ be a Markov string in finite alphabets. For any fixed strongly jointly typical sequence pair $(x^n,y^n)$,
let $Z^n$ be chosen uniformly at random from the set which consists of all sequences that are strongly typical with $y^n$. Let $Q(\cdot)$ be the probability 
measure induced by this random choice. 
Then $\lim_{n\rightarrow \infty}Q((x^n,y^n,Z^n) \mbox{ are not strongly jointly typical})=0$, and the convergence is uniform over the set of strongly jointly typical $(x^n,y^n)$ sequence pairs.
\end{lemma}

\begin{IEEEproof}[Proof of Theorem \ref{theorem:DNSC}]

\noindent\textbf{Proof for the direction $\mathcal{D}_{dis}\supseteq\mathcal{D}^*_{dis}$}: To prove this direction, it suffices to show 
\begin{align}
\mathcal{D}_{dis}\supseteq\bigcup_{\{R_{i,j}:(i,j)\in\mathcal{E}\}:R_{i,j}<\kappa
  C_{i,j}}\mathcal{D}_{dis}(\{R_{i,j}\}_{(i,j)\in\mathcal{E}}).\label{eqn:forwardcondition2}
\end{align}
This is because the achievable distortion region $\mathcal{D}_{dis}$ is closed, and the distortion-matrix-rate function
$\mathcal{D}_{dis}(\{R_{i,j}\}_{(i,j)\in\mathcal{E}})$ is continuous in the relative interior of the non-negative quadrant (implied by its convexity), from which it follows that the
condition $R_{i,j}< \kappa C_{i,j}$ can be replaced by $R_{i,j}\leq
\kappa C_{i,j}$ in (\ref{eqn:forwardcondition2}), implying that $\mathcal{D}_{dis}\supseteq\mathcal{D}^*_{dis}$ where $\mathcal{D}^*_{dis}$ is defined in (\ref{eqn:DisRegionDef}). 

To show (\ref{eqn:forwardcondition2}), let $\epsilon>0$ be some quantity such that 
\begin{align}
\epsilon\leq \min_{(i,j)\in\mathcal{E}} (\kappa C_{i,j}-R_{i,j})\label{eqn:epsilon}
\end{align} 
for any chosen set of $\{R_{i,j},(i,j)\in\mathcal{E}\}$ such that $R_{i,j}<\kappa C_{i,j}$,  $(i,j)\in\mathcal{E}$. For any distortion vector $\{D_{k,j},k\in\mathscr{T}_j\}\in\mathcal{D}_{dis}(\{R_{i,j}\}_{(i,j)\in\mathcal{E}})$, since it is achievable with $\{R_{i,j},(i,j)\in\mathcal{E}\}$, there exists an $l$ such that for any sufficiently large $m$, there exists an $(m,l,\{L_{i,j},(i,j)\in\mathcal{E}\},\{D_{k,j}+\frac{1}{4}\epsilon,k\in\mathscr{T}_j\})$ distributed network source code (see Definition \ref{def:DDNSC}) with a total of $l$ sessions, where 
\begin{align}
R_{i,j}+\frac{\epsilon}{4}\geq \frac{l}{m}\log L_{i,j}. \label{eqn:LandR}
\end{align}
We utilize this source code together with a good channel code for each channel in the original network. 
More precisely, there are at least a total of $\kappa m-1$ channel uses available, and we
shall partition them into $l$ channel sessions, each with at least
$\left\lfloor \frac{\kappa m-1}{l}\right\rfloor$ channel uses. Thus the channel on edge $(i,j)$ in 
each session can support a message of cardinality $\left\lfloor 2^{\left\lfloor\frac{\kappa
  m-1}{l}\right\rfloor (C_{i,j}-\frac{1}{4}\epsilon)}\right\rfloor$, with maximum error probability (among all messages for each channel code) less than $\epsilon$, by choosing $m$ sufficiently large. Each
session of the pure source code has a message output of cardinality no larger than
$L_{i,j}$. Thus
 as long as 
 \begin{align}
 L_{i,j}\leq\left\lfloor 2^{\left\lfloor\frac{\kappa  m-1}{l}\right\rfloor (C_{i,j}-\frac{1}{4}\epsilon)}\right\rfloor, \label{eqn:forwardcondition}
  \end{align}
we can use the digital channel codes to transmit the source code indices with vanishing error probability. It follows that for any $\epsilon>0$, there exists a sufficiently large $m$ such that the total error probability over $l$-sessions is less than $l|\mathcal{E}|\epsilon$ in this network. For (\ref{eqn:forwardcondition}) to hold under the condition (\ref{eqn:LandR}), it suffices to have
\begin{align}
2^{\frac{m}{l} (R_{i,j}+\frac{1}{4}\epsilon)}\leq 2^{(\frac{\kappa m-1}{l}-1) (C_{i,j}-\frac{1}{4}\epsilon)}-1.\label{eqn:forwardcondition1}
\end{align}
Eqn. (\ref{eqn:epsilon}) implies that for any $\epsilon>0$, (\ref{eqn:forwardcondition1}) is true for any sufficiently large $m$ and the fixed $l$ afore-mentioned in the  $(m,l,\{L_{i,j},(i,j)\in\mathcal{E}\},\{D_{k,j}+\frac{1}{4}\epsilon,k\in\mathscr{T}_j\})$ distributed network source code, and subsequently (\ref{eqn:forwardcondition}) holds. Thus for any $\epsilon>0$, by choosing $m$ sufficiently large, the separation based scheme is able to achieve the distortion vector $\{D_{k,j}+\epsilon,k\in\mathscr{T}_j\}$ for any $\{D_{k,j},k\in\mathscr{T}_j\}\in\mathcal{D}_{dis}(\{R_{i,j}\}_{(i,j)\in\mathcal{E}})$  with probability greater than or equal to $(1-l|\mathcal{E}|\epsilon)$, and distortion $D_{\max}$ with probability less than or equal to $l|\mathcal{E}|\epsilon$, for any $R_{i,j}$ such that $\kappa C_{i,j}-R_{i,j}> \epsilon,\, (i,j)\in\mathcal{E}$, where $D_{\max}$ is the maximum distortion value for all the finite-alphabet sources in the network. Since $\epsilon$ can be made arbitrarily small and $D_{\max}$ is finite, and moreover $\mathcal{D}_{dis}$ is a closed set, (\ref{eqn:forwardcondition2}) is indeed true.

\begin{figure*}[tb]
\begin{centering}
\includegraphics[width=15cm]{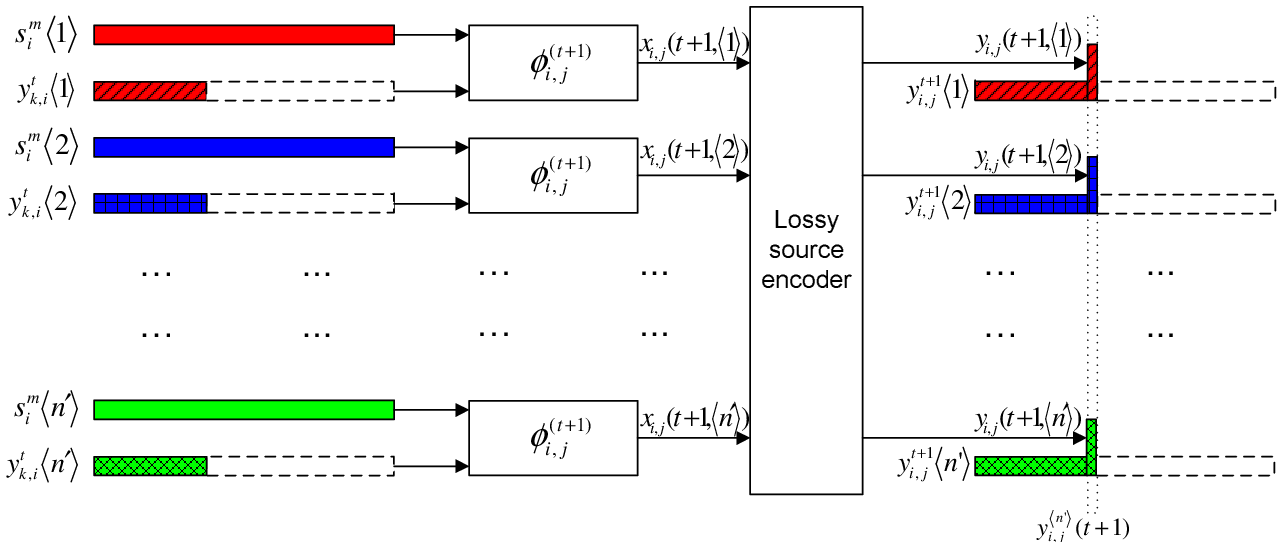}
\caption{Coding operation of $\mathbb{P}_s$ in session $t+1$ for node $i$
  with an incoming link $(k,i)$ and an outgoing link
  $(i,j)$.\label{fig:evolution} Each narrow horizontal box represents a vector; the vectors $y^{t}_{k,i}\langle v\rangle$'s and $y^{t+1}_{i,j}\langle v\rangle$'s are shaded partially because at this point, the later parts have not been generated. Each component of the lossy encoder output, {\em i.e.}, $y_{i,j}(t+1,\langle v\rangle)$, is appended to the existing $y^{t}_{i,j}\langle v\rangle$ to form  $y^{t+1}_{i,j}\langle v\rangle$.}
\end{centering}
\end{figure*}

\noindent\textbf{Proof for the direction $\mathcal{D}_{dis}\subseteq\mathcal{D}^*_{dis}$:} 
We wish to show that if a distortion matrix $\vec{D}$ is achievable in the joint coding problem $\mathbb{P}_j$
(Definitions \ref{def:DNSj1} and \ref{def:DNSj2}), then the rate distortion matrix pair $(\{\kappa C_{e}\},\vec{D})$ 
is also achievable in the source coding problem $\mathbb{P}_{s}$ (Definitions \ref{def:DDNSC} 
and \ref{def:sourcecodingRD}).
For this purpose, we construct an $n$-session distributed network source code for $\mathbb{P}_{s}$ that
operates on a source sequence of length $mn'$ from a joint coding code. 
For any achievable distortion matrix $\vec{D}$  and any $\epsilon>0$, there exists an $(m,n,\{D_{k,j}+\epsilon,k\in\mathscr{T}_j\})$ distributed network joint source-channel code 
(see Definitions \ref{def:DNSj1} and \ref{def:DNSj2}), where $n\leq\kappa m$. Let us fix this joint coding code, and use it to construct a source code for $\mathbb{P}_s$.

First partition the source sequence $S^{mn'}_i$, $i=1,2,\ldots,N$, into $n'$
disjoint block components, each of length $m$. The
$v$-th block component of $S^{mn'}_i$ is written as $S^m_i\langle
v\rangle$, {\em i.e.},
\begin{align*}
&S^m_i\langle v\rangle\triangleq\bigg{(}S_i((v-1)m+1),S_i((v-1)m+2),\nonumber\\
&\qquad\qquad\qquad\qquad\qquad\qquad\qquad\qquad\ldots,S_i(vm)\bigg{)},\nonumber\\
&\qquad\qquad\qquad\quad\qquad\qquad\qquad v=1,2,\ldots,n\rq{}.
\end{align*}
To make this partition explicit, $S^{mn'}_i$ is written in the sequel as $S^{m,\langle n'\rangle}_i$.

\textbf{Codebook generation:} For each $(i,j)\in\mathcal{E}$ and each {\em session} 
$t=1,2,\ldots,n$, a source coding codebook
$\mathcal{C}_{(i,j),t}$ of size $2^{n'(I(X_{i,j}(t);Y_{i,j}(t))+\delta)}$ is generated
by choosing from the strongly typical set of the random variable $Y_{i,j}(t)$ uniformly at random with replacement, where $\delta>0$ is a small quantity $\delta\rightarrow 0$ as $n\rq{}\rightarrow \infty$. 
This codebook is revealed to both the encoder and the
decoder on edge $(i,j)$ in the problem $\mathbb{P}_s$. 

\textbf{Encoding and decoding:} For session $t=1$ at any given edge
$(i,j)\in\mathcal{E}$, we first apply the chosen joint source
channel encoding function $\phi^{(1)}_{i,j}$ on each block component
$S^m_i\langle v\rangle$, $v=1,2,\ldots,n'$; denote the output $\phi^{(1)}_{i,j}(s^m_i\langle v\rangle)$ as $x_{i,j}(1,\langle v\rangle)$. The following length-$n'$ vector is formed 
by concatenating them 
\begin{align}
x^{\langle n'\rangle}_{i,j}(1)\triangleq (x_{i,j}(1,\langle 1\rangle),x_{i,j}(1,\langle 2\rangle),\ldots,x_{i,j}(1,\langle n'\rangle)).
\end{align}

For each $(i,j)\in\mathcal{E}$, if $x^{\langle n'\rangle}_{i,j}(1)$ is strongly typical, we find a codeword $y^{\langle n'\rangle}_{i,j}(1)$ in
$\mathcal{C}_{(i,j),1}$ such that $x^{\langle n'\rangle}_{i,j}(1)$ and $y^{\langle n'\rangle}_{i,j}(1)$ are
strongly jointly typical with respect to $P(X_{i,j}(1),Y_{i,j}(1))$; if there does not exist such a codeword, an error is declared. Denote the index of this chosen $y^{\langle n'\rangle}_{i,j}(1)$ codeword as $w_{i,j}(1)$; the $v$-th location in the vector $y^{\langle n'\rangle}_{i,j}(1)$ is written as $y_{i,j}(1,\langle v\rangle)$.
The encoding functions $\tilde{\phi}^{(1)}_{i,j}$ for $\mathbb{P}_s$ are given by 
\begin{align}
\tilde{\phi}^{(1)}_{i,j}\left(s^{m,\langle n'\rangle}_i\right)=w_{i,j}(1),\qquad (i,j)\in\mathcal{E}.
\end{align}

In the $t$-th session, for any given edge
$(i,j)\in\mathcal{E}$, the chosen joint source-channel
encoding function $\phi^{(t)}_{i,j}$ is applied, and the outputs are concatenated (see Fig. \ref{fig:evolution}), {\em i.e.},
\begin{align}
x^{\langle n'\rangle}_{i,j}(t)=&\bigg{(}\phi^{(t)}_{i,j}(s^m_i\langle 1\rangle,\{y^{t-1}_{k,i}\langle 1\rangle,(k,i)\in\mathcal{E}\}),\nonumber\\
&\quad\phi^{(t)}_{i,j}(s^m_i\langle 2\rangle,\{y^{t-1}_{k,i}\langle 2\rangle,(k,i)\in\mathcal{E}\}),\nonumber\\
&\quad\ldots,\phi^{(t)}_{i,j}(s^m_i\langle n'\rangle,\{y^{t-1}_{k,i}\langle n'\rangle,(k,i)\in\mathcal{E}\})\bigg{)},
\end{align}
where
\begin{align}
y^{t-1}_{i,j}{\langle v\rangle}\triangleq(y_{i,j}(1,\langle v\rangle),y_{i,j}(2,\langle v\rangle),\ldots,y_{i,j}(t-1,\langle v\rangle)).
\end{align}
For any $(i,j)\in\mathcal{E}$, if 
$x^{\langle n'\rangle}_{i,j}(t)$ is strongly typical, find a codeword $y^{\langle n'\rangle}_{i,j}(t)$ in $\mathcal{C}_{(i,j),t}$ such that
$x^{\langle n'\rangle}_{i,j}(t)$ and $y^{\langle n'\rangle}_{i,j}(t)$ are strongly jointly typical with
respect to $P(X_{i,j}(t),Y_{i,j}(t))$; if there does not exist such a codeword, an error is declared. The index of the chosen codeword $y^{\langle n'\rangle}_{i,j}(t)$ in $\mathcal{C}_{(i,j),t}$ is denoted as $w_{i,j}(t)$, and thus the encoding functions $\tilde{\phi}^{(t)}_{i,j}$ for $\mathbb{P}_s$ are 
\begin{align}
\tilde{\phi}^{(t)}_{i,j}\left(s^{m,\langle n'\rangle}_i,\{w^{t-1}_{k,i},(k,i)\in\mathcal{E}\}\right)=w_{i,j}(t),\qquad (i,j)\in\mathcal{E}.
\end{align}

After $n$ sessions of encoding, at node $j\in\mathcal{V}$,
the chosen joint source-channel decoding function
$\psi_{k,j}$ is applied to reconstruct the $v$-th block component of source
$s^{m,\langle n'\rangle}_k$, $k\in\mathscr{T}_j$, {\em i.e.},
\begin{align}
&\hat{s}^{m}_{k,j}\langle v\rangle=\psi_{k,j}(s^{m}_j\langle v\rangle,\{y^n_{i,j}\langle v\rangle,(i,j)\in\mathcal{E}\}), \nonumber\\
&\qquad\qquad\qquad\qquad\qquad v=1,2,\ldots,n',
\end{align}
which are then concatenated to form $\hat{s}^{m,\langle n'\rangle}_{k,j}$, {\em i.e.}, the length-$mn'$ reconstruction of source $k$ at node $j$. Thus the decoding functions $\tilde{\psi}_{k,j}$ for $\mathbb{P}_s$ are given as 
\begin{align}
\tilde\psi_{k,j}(s^{m,\langle n'\rangle}_j,\{w^{n}_{i,j},(i,j)\in\mathcal{E}\})=\hat{s}^{m,\langle n'\rangle}_{k,j},\quad k\in\mathscr{T}_j.
\end{align}

\textbf{Error probability and distortion analysis:} There are three kinds of error events in session-$t$
\begin{itemize}
\item $E^{(1)}_{t}$: $(s^{m,\langle n'\rangle}_{\mathcal{V}},x^{t,\langle n'\rangle}_{\mathcal{E}},y^{t-1,\langle n'\rangle}_{\mathcal{E}})$ are not strongly jointly typical with respect to $P(S^m_{\mathcal{V}},X^{t}_{\mathcal{E}},Y^{t-1}_{\mathcal{E}})$; 
\item $E^{(2)}_{t,(i,j)}$: for an edge $(i,j)\in\mathcal{E}$, given $x^{\langle n'\rangle}_{i,j}(t)$ is strongly typical, there does not exist any codeword in
$\mathcal{C}_{(i,j),t}$ such that it is strongly jointly typical with $x^{\langle n'\rangle}_{i,j}(t)$ with respect to $P(X_{i,j}(t),Y_{i,j}(t))$;
\item $E^{(3)}_{t}$: $(s^{m,\langle n'\rangle}_{\mathcal{V}},x^{t,\langle n'\rangle}_{\mathcal{E}},y^{t-1,\langle n'\rangle}_{\mathcal{E}})$ and $y^{\langle n'\rangle}_{\mathcal{E}}(t)$ are not strongly jointly typical with respect to
$P(S^m_{\mathcal{V}},X^t_{\mathcal{E}},Y^t_{\mathcal{E}})$.
\end{itemize}

Note $E^{(3)}_0$ is the event that $s^{m,\langle n'\rangle}_{\mathcal{V}}$ is not strongly jointly typical. The overall error event is given as 
\begin{align}
E_{n'}=&\bigcup_{t=1}^{n}(E^{(1)}_{t} \cup \bigcup_{(i,j)\in\mathcal{E}} E^{(2)}_{t,(i,j)}\cup E^{(3)}_{t})\nonumber\\
=&\bigcup_{t=1}^{n}\bigg{(}\overline{E^{(3)}_{t-1}}\cap E^{(1)}_{t}\bigg{)} \cup \bigg{(}\overline{E^{(1)}_{t}} \cap \bigcup_{(i,j)\in\mathcal{E}} E^{(2)}_{t,(i,j)}\bigg{)} \nonumber\\
&\qquad\qquad\cup \bigg{(}\overline{E^{(1)}_t \cup \bigcup_{(i,j)\in\mathcal{E}} E^{(2)}_{t,(i,j)}}\cap E^{(3)}_{t}\bigg{)},
\end{align}
where $\overline{S}$ is the complement of $S$. 

By the union bound, we have
\begin{align}
\mbox{Pr}(E_{n'})\leq& \sum_{t=1}^n \mbox{Pr}(\overline{E^{(3)}_{t-1}}\cap E^{(1)}_{t})\nonumber\\
&\,+\sum_{t=1}^n \mbox{Pr}\bigg{(}\overline{E^{(1)}_{t}} \cap \bigcup_{(i,j)\in\mathcal{E}} E^{(2)}_{t,(i,j)}\bigg{)}\nonumber\\
&\quad+\sum_{1}^n \mbox{Pr}\bigg{(}\overline{E^{(1)}_t \cup \bigcup_{(i,j)\in\mathcal{E}} E^{(2)}_{t,(i,j)}}\cap E^{(3)}_{t}\bigg{)}\label{eqn:errorevents}.
\end{align}
Next we show that $\mbox{Pr}(E_{n'})\rightarrow 0$ as $n'\rightarrow \infty$. Firstly, $\mbox{Pr}(E^{(3)}_0)\rightarrow 0$ by the basic properties of the strongly jointly typical sequences (\cite{CoverThomas}, pp. 358-362).
Since $x^{\langle n'\rangle}_{\mathcal{E}}(1)$ is a deterministic function of $s^{m,\langle n'\rangle}_{\mathcal{V}}$, $\mbox{Pr}(\overline{E^{(3)}_{0}}\cap E^{(1)}_{1})\rightarrow 0$, and similarly 
$\mbox{Pr}(\overline{E^{(3)}_{t-1}}\cap E^{(1)}_{t})\rightarrow 0$ for $t=2,3,\ldots,n$. For the second summation in (\ref{eqn:errorevents}), 
\begin{align}
&\sum_{t=1}^n \mbox{Pr}\bigg{(}\overline{E^{(1)}_{t}} \cap \bigcup_{(i,j)\in\mathcal{E}} E^{(2)}_{t,(i,j)}\bigg{)}\nonumber\\
&\qquad\qquad\qquad\leq \sum_{t=1}^n \sum_{(i,j)\in\mathcal{E}}\mbox{Pr}(\overline{E^{(1)}_{t}} \cap E^{(2)}_{t,(i,j)}),
\end{align}
by the union bound.
Since $\overline{E^{(1)}_{t}}$ implies that $x^{\langle n'\rangle}_{i,j}(t)$ is strongly typical, $\mbox{Pr}(\overline{E^{(1)}_{t}} \cap E^{(2)}_{t,(i,j)})\rightarrow 0$ for any $t$ and $(i,j)\in \mathcal{E}$, by the properties of the strongly typical sequences (\cite{CoverThomas}, Lemma 13.6.2), and the fact that the number of codewords in $\mathcal{C}_{(i,j),t}$ is $2^{n'(I(X_{i,j}(t);Y_{i,j}(t))+\delta)}$. 

To bound the third summation in (\ref{eqn:errorevents}), let us fix an arbitrary order for the edges in the set $\mathcal{E}$, and write it as $e_1,e_2,\ldots,e_{|\mathcal{E}|}$. Define $E^{(3)}_{t,k}$ as the event that $(s^{m,\langle n'\rangle}_{\mathcal{V}},x^{t,\langle n'\rangle}_{\mathcal{E}},y^{t-1,\langle n'\rangle}_{\mathcal{E}})$ and $(y^{\langle n'\rangle}_{e_1}(t),y^{\langle n'\rangle}_{e_2}(t),\ldots,y^{\langle n'\rangle}_{e_k}(t))$ are not strongly jointly typical.
We can then rewrite
\begin{align}
&\overline{E^{(1)}_t \cup \bigcup_{(i,j)\in\mathcal{E}} E^{(2)}_{t,(i,j)}}\cap E^{(3)}_{t}\nonumber\\
&=\bigcup_{k=1}^{|\mathcal{E}|} \overline{E^{(1)}_t \cup \bigcup_{(i,j)\in\mathcal{E}} E^{(2)}_{t,(i,j)}\cup E^{(3)}_{t,k-1}}\cap E^{(3)}_{t,k}\triangleq  \bigcup_{k=1}^{|\mathcal{E}|}E^{(3)*}_{t,k},
\end{align}
where $E^{(3)}_{t,0}\triangleq\emptyset$. To bound $\mbox{Pr}(E^{(3)*}_{t,k})$, observe that 
\begin{align}
&\left(S^{m}_{\mathcal{V}},X^{t-1}_{\mathcal{E}},X_{\mathcal{E}\setminus e_k}(t),Y^{t-1}_{\mathcal{E}},Y_{e_1,e_2,\ldots,e_{k-1}}(t)\right)\nonumber\\
&\qquad\qquad\qquad\qquad\leftrightarrow X_{e_{k}}(t) \leftrightarrow Y_{e_{k}}(t)
\end{align}
is a Markov string. 
Invoking Lemma \ref{lemma:markov} gives that $\mbox{Pr}(E^{(3)*}_{t,k})\rightarrow 0$, for any $t$ and $k$, as $n'\rightarrow \infty$.

There are a total of $n$ terms in the first summation of (\ref{eqn:errorevents}), $n|\mathcal{E}|$ terms in the second, and $n|\mathcal{E}|$ terms in the third. Since $n$ and $|\mathcal{E}|$ are fixed here, and each term can be made arbitrarily small by making $n'$ sufficiently large, we have $\mbox{Pr}(E_{n'})\rightarrow 0$ as $n'\rightarrow \infty$. This implies that the sequences 
$(s^{m,\langle n'\rangle}_{\mathcal{V}},x^{m,\langle n'\rangle}_{\mathcal{E}},y^{m,\langle n'\rangle}_{\mathcal{E}})$
are strongly jointly typical with respect to the original distribution
$P(S^{m}_{\mathcal{V}},X^n_{\mathcal{E}},Y^n_{\mathcal{E}})$ with probability arbitrarily close to one as $n'\rightarrow \infty$.
This further implies that $s^{m,\langle n'\rangle}_{k}$ and $\hat{s}^{m,\langle n'\rangle}_{k,j}$ are strongly jointly typical with respect to $P(S^m_k,\hat{S}^m_{k,j})$, and the new code induces a distortion $D_{k,j}+\epsilon+\delta'$, where $\delta'\rightarrow 0$ as $n'\rightarrow \infty$.

\textbf{Rate analysis:} In the chosen $(m,n,\{D_{k,j}+\epsilon,k\in\mathscr{T}_j\})$ joint source-channel code for $\mathbb{P}_j$, for each link $e=(i,j)\in\mathcal{E}$, the conventional   channel coding theorem implies
\begin{align}
I(X_e(t);Y_e(t))\leq C_{e},\quad t=1,2,\ldots,n,
\end{align}
where $C_e$ is the capacity of the channel on edge $e$.  
Thus the cardinality of above-constructed source code for $\mathbb{P}_s$ in each 
session associated with any given link $e$ is bounded as 
\begin{align}
2^{n'(I(X_e(t);Y_e(t))+\delta)}\leq 2^{n'(C_e+\delta)},\quad t=1,2,\ldots,n.
\end{align}
It follows according to Definition \ref{def:sourcecodingRD} that the following rate is achievable in problem $\mathbb{P}_s$
\begin{align}
R_{e}=\frac{n}{mn'}\log2^{n'(C_e+\delta)}\leq \kappa (C_e+\delta).
\end{align}

\textbf{Finishing the proof $\mathcal{D}_{dis}\subseteq\mathcal{D}^*_{dis}$:} We have shown that by utilizing a chosen joint source-channel code $(m,n,\{D_{k,j}+\epsilon,k\in\mathscr{T}_j\})$ for $\mathbb{P}_j$, the constructed sequence of source codes for $\mathbb{P}_s$ can operate at the rate-distortion-matrix tuple $(\{R_{i,j}=\kappa (C_{i,j}+\delta), (i,j)\in \mathcal{E}\},\vec{D}+\epsilon+\delta')$, where $\delta$ and $\delta'$ can be made arbitrarily small by letting $n'\rightarrow \infty$. 
Since the achievable rate-distortion-matrix region for $\mathbb{P}_s$ is a closed set, the tuple $(\{R_{i,j}=\kappa C_{i,j},(i,j)\in \mathcal{E}\},\vec{D}+\epsilon)$ is achievable in $\mathbb{P}_s$. Since the distortion matrix $\vec{D}$ is achievable in $\mathbb{P}_j$, for any $\epsilon>0$, there exists an $(m,n,\{D_{k,j}+\epsilon,k\in\mathscr{T}_j\})$ joint source-channel code, where $n\leq\kappa m$ by choosing $n$ sufficiently large. Thus for any $\epsilon>0$, $(\{R_{i,j}=\kappa C_{i,j},(i,j)\in \mathcal{E}\},\vec{D}+\epsilon)$ is achievable in $\mathbb{P}_s$. Again by the fact that the achievable rate-distortion-matrix region for $\mathbb{P}_s$ is closed, the tuple $(\{R_{e}=\kappa C_{e},e\in \mathcal{E}\},\vec{D})$ is achievable for $\mathbb{P}_s$.  Applying (\ref{eqn:DisRegionDef}) now completes the proof for $\mathcal{D}_{dis}\subseteq\mathcal{D}^*_{dis}$.  
\end{IEEEproof}

\section{Optimality of Separation for Joint Source-Channel Multiple Unicast with Distortions}
\label{sec:proofunicast}

The following theorem formally states that source-channel separation is optimal in the JSCMUD problem. Recall $\mathcal{D}_{uni}$ and $\mathcal{D}^*_{uni}$ given in Definition \ref{def:JScodeUnicastRegion} and Eqn. (\ref{eqn:Duni}), respectively.
\begin{theorem}
\label{theorem:main}
$\mathcal{D}_{uni}=\mathcal{D}^*_{uni}$.
\end{theorem}

\begin{IEEEproof}[Proof of Theorem \ref{theorem:main}]
The direction $\mathcal{D}_{uni}\supseteq\mathcal{D}^*_{uni}$ is
rather obvious except one technicality. The channel coding problem given in 
Definitions \ref{def:DigChannelcodeUnicast} and \ref{def:diguniregion} has an error probability defined as averaged over all messages. However, the codeword indices for the source codes may not have a uniform distribution, and thus the overall error probability by combing the source code and the channel code may be larger if the mapping between the source code indices and the channel code indices are chosen poorly. This however can be resolved using a standard random coding argument \cite{CoverThomas} over all possible one-to-one mappings, and the detail is thus omitted.    

We next focus on the other
direction $\mathcal{D}_{uni}\subseteq\mathcal{D}^*_{uni}$.  For any
achievable distortion vector $(D_1,D_2,\ldots,D_M)$, and any $\epsilon>0$,
there exists an $(m,n,D_1+\epsilon,D_2+\epsilon,\ldots,D_M+\epsilon)$ JSCMUD code, 
where $n\leq \kappa m$ (see Definitions \ref{def:JScodeUnicast} and \ref{def:JScodeUnicastRegion}). 
The sources and the above given block code induce a joint distribution
\begin{align}
\label{eqn:superdistribution}
\prod_{i=1}^MP(S^m_i)\cdot P\bigg{(}\hat{S}^m_1,\hat{S}^m_2,\ldots,\hat{S}^m_M\bigg{|}S^{m}_1,S^{m}_2,\ldots,S^{m}_M\bigg{)},
\end{align}
and the second term can be viewed as the transition probability of a
block-level interference channel, which has input alphabets
$\mathcal{S}^{m}_1\times\mathcal{S}^{m}_2\times\ldots\times\mathcal{S}^{m}_M$, and output
alphabets $\hat{\mathcal{S}}^{m}_1\times\hat{\mathcal{S}}^{m}_2\times\ldots\times\hat{\mathcal{S}}^{m}_M$. Moreover,
by the conventional rate-distortion theorem \cite{CoverThomas}, 
\begin{align}
\label{eqn:mutualinformation}
I(S^{m}_i;\hat{S}^{m}_i)\geq m
R_{i}(D_i+\epsilon),\qquad i=1,2,\ldots,M,
\end{align}
where $R_i(\cdot)$ is the rate-distortion function for source $S_i$. This super interference 
channel operates in the same manner as a memoryless interference channel, however on a block level $(S^{m}_1,S^{m}_2,\ldots,S^{m}_M)\rightarrow (\hat{S}^m_1,\hat{S}^m_2,\ldots,\hat{S}^m_M)$, 
instead of on a single time instance level $(X_1,X_2,\cdots,X_N)\rightarrow (Y_1,Y_2,\cdots,Y_N)$.

Next we show that if a distortion vector $(D_1,D_2,\ldots,D_M)$ is achievable on the joint coding problem $\mathbb{P}_j$ (Definitions \ref{def:JScodeUnicast} and \ref{def:JScodeUnicastRegion}), then the rate vector $(R_1(D_1),R_1(D_1),\ldots,R_M(D_M))$ 
is achievable on the channel coding problem   $\mathbb{P}_c$ (Definitions \ref{def:DigChannelcodeUnicast} and \ref{def:diguniregion}).
For this purpose, we construct a multiple unicast channel code for $\mathbb{P}_c$ 
using the afore-mentioned $(m,n,D_1+\epsilon,D_2+\epsilon,\ldots,D_M+\epsilon)$ joint source-channel code
for $\mathbb{P}_j$. The coding scheme for $\mathbb{P}_c$ can be formally described as follows.

\textbf{Codebook generation:} For each source $S_i$, $2^{mn' (R_{i}(D_i+\epsilon)-\delta)}$ codewords 
of length-$(mn')$ are generated independently, according to the $mn'$-th product distribution of $P(S_i)$; 
denote this codebook as $\mathcal{C}_i$. The codebooks are revealed to 
all the nodes.

\textbf{Encoding:} To encode for $\mathbb{P}_c$, for a message $w_i$, choose the $w_i$-th codeword $s^{mn'}_i(w_i)$ in the $\mathcal{C}_i$ codebook generated above. Each codeword is partitioned into $n'$ blocks of equal length, and denote the $v$-th block as $s^m_i(w_i,\langle v\rangle)$; to emphasize this partition, we also write $s^{mn'}_i(w_i)$ as $s^{m,\langle n'\rangle}_i(w_i)$. 
For a fixed $v$, the blocks $(s^m_1(w_1,\langle v\rangle),s^m_2(w_2,\langle v\rangle),\ldots,s^m_M(w_M,\langle v\rangle))$ from the chosen codewords at all the nodes can be viewed as the length-$m$ source vectors in $\mathbb{P}_j$, and thus the chosen $(m,n,D_1+\epsilon,D_2+\epsilon,\ldots,D_M+\epsilon)$ JSCMUD encoding functions and decoding functions can be used on them. This results in a set of reconstruction sequences $(\hat{s}^m_1\langle v\rangle,\hat{s}^m_2\langle v\rangle,\ldots,\hat{s}^m_M\langle v\rangle)$. At the end of $n'$ blocks, we concatenate the reconstruction for each source block as $\hat{s}^{mn'}_i=\hat{s}^{m,\langle n'\rangle}_i=(\hat{s}^{m}_i\langle 1\rangle,\hat{s}^{m}_i\langle 2\rangle,\ldots,\hat{s}^{m}_i\langle n'\rangle)$.

Mathematically, let the chosen joint source-channel encoding function and decoding function at node $j$ be $\phi^{(t)}_{j}$ and $\psi_{k,j}$, respectively. Similarly as the notation of $s^m_i(w_i,\langle v\rangle)$, the $v$-th length-$n$ block of $y^{n,\langle n'\rangle}_i$ is written as $y^n_i\langle v\rangle$, and the first $t$ symbols of the block $y^n_i\langle v\rangle$ is written as $y^t_i\langle v\rangle$.
Then the new channel code encoding function $\tilde{\phi}^{(t')}_{j}$ is given by
\begin{align}
&\tilde{\phi}^{((v-1)n+t)}_{j}\bigg{(}\{w_i,i\in\mathscr{S}_j\},y^{(v-1)n+t-1}_j\bigg{)}\nonumber\\
&\qquad\qquad=\phi^{(t)}_{j}\bigg{(}\{s^m_i(w_i,\langle v\rangle),i\in\mathscr{S}_j\},y^{t-1}_j\langle v\rangle\bigg{)},\nonumber\\
&\qquad\qquad\qquad\qquad v=1,2,\ldots,n',\quad t=1,2,\ldots,n.
\end{align}
The reconstructions are 
\begin{align}
&\hat{s}^{m}_j\langle v\rangle = \psi_{k,j}\bigg{(}\{s^m_i(w_i,\langle v\rangle),i\in\mathscr{S}_j\},y^{n}_j\langle v\rangle\bigg{)},\quad k\in\mathscr{T}_j.
\end{align}

\textbf{Decoding:} At node $j$, for which $k\in \mathscr{S}_j$, find a unique codeword in the codebook $\mathcal{C}_k$ such that it is (weakly) jointly typical \cite{CoverThomas} with $\hat{s}^{m,\langle n'\rangle}_k$ according to the distribution $P(S^{m}_k,\hat{S}^{m}_k)$, {\em i.e.}, the marginal from (\ref{eqn:superdistribution}). If there is a unique one, the corresponding message $w^*_k$ is declared; otherwise an error is declared.

\textbf{Error probability analysis:} There are three kinds of errors 
\begin{itemize}
\item $E^{(1)}$: $(s^{m,\langle n'\rangle}_1(w_1),s^{m,\langle n'\rangle}_2(w_2),\ldots,s^{m,\langle n'\rangle}_M(w_M))$ are not jointly typical with respect to (\ref{eqn:superdistribution});
\item $E^{(2)}$: $(s^{m,\langle n'\rangle}_1(w_1),s^{m,\langle n'\rangle}_2(w_2),\ldots,s^{m,\langle n'\rangle}_M(w_M),$ $\hat{s}^{m,\langle n'\rangle}_1,\hat{s}^{m,\langle n'\rangle}_2,\ldots,\hat{s}^{m,\langle n'\rangle}_M)$ are not jointly typical with respect to (\ref{eqn:superdistribution}); 
\item $E^{(3)}_i$: for a given message $w_i$, there is more than one codeword in $\mathcal{C}_i$ that is jointly typical with $\hat{s}^{m,\langle n'\rangle}_i(w_i)$, with respect to the marginal of (\ref{eqn:superdistribution}).
\end{itemize}

By the union bound, the overall error probability can be bounded as
\begin{align}
&\mbox{Pr}(E_{n'})\leq \mbox{Pr}(E^{(1)})+\mbox{Pr}(\overline{E^{(1)}}\cap E^{(2)})\nonumber\\
&\qquad\qquad\qquad+\sum_{i=1}^M\mbox{Pr}(\overline{E^{(2)}}\cap E^{(3)}_i).
\label{eqn:JSCMUDerror}
\end{align}

Since all the codewords are generated according to $P(S_i)$'s independently, by the basic properties of the jointly typical sequences (\cite{CoverThomas}, Theorem 14.2.1), $\mbox{Pr}(E^{(1)})\rightarrow 0$ as $n'\rightarrow\infty$. This implies that
the reconstructions $\{\hat{s}^{m,\langle n'\rangle}_i,i=1,2,\ldots,M\}$ are jointly typical with $\{s^{m,\langle n'\rangle}_i(w_i),i=1,2,\ldots,M\}$ with probability approaching one, {\em i.e.}, $\mbox{Pr}(\overline{E^{(1)}}\cap E^{(2)})\rightarrow 0$ as $n'\rightarrow\infty$. It follows that $\mbox{Pr}(\overline{E^{(2)}}\cap E^{(3)}_i)\rightarrow0$ as $n'\rightarrow\infty$, by (\ref{eqn:mutualinformation}) and the basic property of the jointly typical sequences (\cite{CoverThomas}, Theorem 14.2.1 and Theorem 14.2.2), and the fact that the number of codewords in $\mathcal{C}_i$ is $2^{mn' (R_{i}(D_i+\epsilon)-\delta)}$. Since there are a total of $M+2$ terms in (\ref{eqn:JSCMUDerror}), 
$\mbox{Pr}(E_{n'})\rightarrow0$ as $n'\rightarrow\infty$.

\textbf{Finishing the proof $\mathcal{D}_{uni}\subseteq\mathcal{D}^*_{uni}$:} We have shown that by fixing a joint source-channel code $(m,n,D_1+\epsilon,D_2+\epsilon,\ldots,D_M+\epsilon)$ for $\mathbb{P}_j$, the constructed sequence of channel codes can operate at rate tuple $(R_1(D_1+\epsilon)-\delta,R_2(D_2+\epsilon)-\delta,\ldots,R_M(D_M+\epsilon)-\delta)$ for $\mathbb{P}_c$, where $\delta$ can be made arbitrarily small by letting $n'\rightarrow \infty$. Since the set $\mathcal{C}_{uni}$ is closed, the rate tuple $(R_1(D_1+\epsilon),R_2(D_2+\epsilon),\ldots,R_M(D_M+\epsilon))\in \mathcal{C}_{uni}$. Since the rate-distortion functions $R_i(\cdot)$'s are continuous and the capacity region $\mathcal{C}_{uni}$ is closed, we have $(R_1(D_1),R_2(D_2),\ldots,R_M(D_M))\in \mathcal{C}_{uni}$. It follows that $(D_1,D_2,\ldots,D_M)\in \mathcal{D}^*_{uni}$ by the definition of $\mathcal{D}^*_{uni}$ in (\ref{eqn:Duni}), and thus $\mathcal{D}_{uni}\subseteq\mathcal{D}^*_{uni}$.  This completes the proof.
\end{IEEEproof}

\section{Approximate Optimality of Separation for Joint Source-Channel Multiple Multicast with Distortions}
\label{sec:proofmulticast}

In this section the third scenario where there could be multiple receivers
interested in the same source at different distortion levels is examined. 
We limit ourselves to a set of distortion measures referred to as the
``difference" distortion measures, whose properties play an important role in the proof. 
More precisely, $\hat{\mathcal{X}}=\mathcal{X}$ in this class of distortion measures, where $\mathcal{X}$ is an Abelian group with a proper addition operation; furthermore,
the distortion mapping $d(x,\hat{x})$ is a function of $x-\hat{x}$, and we shall write it as $d(x,\hat{x})=d(x-\hat{x})$.

Some necessary definitions are quoted next from \cite{Zamir:96}. For random
variables $N$ and $X$ in the alphabet $\mathcal{X}$, the capacity of the
additive noise channel $X\rightarrow X + N$, under a $d(\cdot)$ distortion constraint is defined as
\begin{align}
\label{eq:Cdef}
C(D, N) = \sup_{X:X\bot N, \mathbb{E}{d(X)}\leq D} I(X; X + N).
\end{align}
The addition $+$ is in the Abelian group $\mathcal{X}$ (e.g., real addition, modulo addition or finite 
field addition), and $\bot$ stands for
independence. 
The minimax (or worst noise) capacity is defined as
\begin{align}
\label{eq:DefCxD}
C_{\mathcal{X}}(D)=\inf_{\mathbb{E}{d(N)}\leq D} C(D,N).
\end{align}
$C_{\mathcal{X}}(D)$ can be interpreted as the capacity at equilibrium
in a mutual information jammer game, played over an
additive-noise channel, in which both the expected noise and expected input are
limited to within $D$ in terms of $d(\cdot)$. The quantity
$C_{\mathcal{X}}(D)$ is a function of $D$ in general, however simplification is possible in some cases. 
Particularly, when the
distortion is the mean squared error, $C_{\mathcal{X}}(D)$ is always $0.5$
bit \cite{Zamir:96}.

Our approximation result is in a genie-aided form, 
where additional communication links with bounded 
capacities are provided by a genie. 
We show that a separation-based approach using the original 
communication network together with the additional 
genie-provided communication links can achieve any 
distortion matrix $\vec{D}$ that is achievable in the 
original communication network 
with arbitrary joint coding schemes. 
It will become clear in the proof that if the reconstructions of $S_i$ at multiple destinations in the set $\mathscr{Q}_i$ are 
required to be at the same distortion level a priori, 
then these destinations can be viewed as a single super-destination, and the problem can be 
reduced; therefore, without loss of generality  they are assumed to be at different distortion levels.

The decreasing sequence of distortions for the elements on the $i$-th row in the distortion matrix specifies 
an order $O_i$ of the set $\mathscr{Q}_i$; let $O_i(j)$ be the $j$-th element 
in the set of $\mathscr{Q}_i$ according to the order $O_i$. 
We require these genie-provided links to support degraded message set broadcast from source $S_i$ to the nodes in the set $\mathscr{Q}_i$ for each $i$ where $|\mathscr{Q}_i|>1$: for such a source $S_i$, for each $j\in \mathcal{I}_{|\mathscr{Q}_i|}$, there is a common link of capacity $R_{i,O_i(j)}$ per source sample\footnote{If $S_i$ is present at more than one node, {\em i.e.},  $|\{k:i\in \mathscr{S}_k\}|>1$, then $R_{i,O_i(j)}$ should be the sum rate per source sample of such common links from each of the node in $\{k:i\in \mathscr{S}_k\}$ to all the nodes $O_i(j),O_i(j+1),\ldots,O_i(|\mathscr{Q}_i|)$.} from $S_i$ to all the nodes $O_i(j),O_i(j+1),\ldots,O_i(|\mathscr{Q}_i|)$. These rate entries are collected and written together as the rate matrix $\vec{R}$. Consider adding these genie-provided links on top of the original source communication network, and denote the achievable distortion using a separation approach of successive refinement coupled with superposition channel code on this new communication network as $\mathcal{D}^{**}_{mul}(\vec{O},\vec{R})$. 

\textbf{Example:} Consider the example given in Fig. \ref{fig:example2}. The sets $\mathscr{Q}_i$'s are 
\begin{align}
\mathscr{Q}_1=\{3,4\},\quad \mathscr{Q}_2=\{4\},\quad \mathscr{Q}_3=\{3\}.
\end{align}
The orders when the distortion of $\hat{S}_{1,3}$ is larger than $\hat{S}_{1,4}$ are 
\begin{align}
O_1=(3,4),\quad O_2=(4),\quad O_3=(3).
\end{align}
The rate matrix of the genie-provided links has the form
\begin{align}
\vec{R}=\left[
\begin{array}{cccc}
\Diamond&\Diamond&R_{1,3}&R_{1,4}\\
\Diamond&\Diamond&\Diamond&\Diamond\\
\Diamond&\Diamond&\Diamond&\Diamond
\end{array}
\right]
\end{align}
where $\Diamond$ at row-$i$ and column-$j$ means that the genie does not provide any additional communication capability from source $S_i$ to node $j$, thus $R_{i,j}$ is not defined. The new network consisting of the additional genie-provided links on top of the original source communication network is given in Fig. \ref{fig:genienew}.

\begin{figure}[tb]
\begin{centering}
\includegraphics[width=8cm]{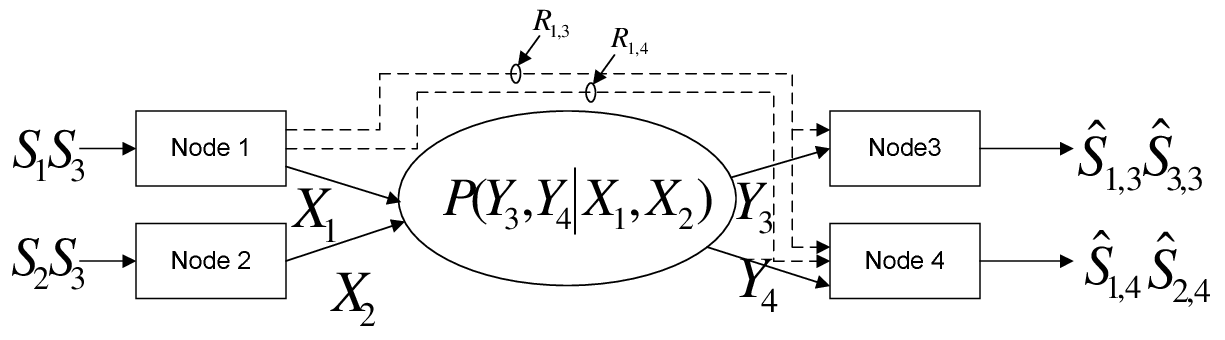}
\caption{\label{fig:genienew} The example in Fig. \ref{fig:example2} with the additional genie-provided links, which are drawn in dashed lines. The region $\mathcal{D}^{**}_{mul}(\vec{O},\vec{R})$ is the achievable distortion region using a separation-based scheme on this joint network.}
\end{centering}
\end{figure}

The following theorem is our first result on general network multicast.

\begin{theorem}
\label{theorem:multicast}
Let $\vec{D}$ be an achievable distortion matrix by joint source-channel coding, for which $\vec{O}$ is the corresponding orders induced by $\vec{D}$. For any random variable $U_{i,j}$ in the Abelian group $\mathcal{X}_i$, $j=1,2,\ldots,|\mathscr{Q}_i|$, such that
\begin{align}
&U_{i,O_i(|\mathscr{Q}_i|)}=V_{i,O_i(|\mathscr{Q}_i|)}\\
&U_{i,O_i(j)}=U_{i,O_i(j+1)}+V_{i,O_i(j)},
\end{align}
where $O_i(j)$ is the $j$-th node index in the set of $\mathscr{Q}_i$ according to the order $O_i$, and $V_{i,j}$'s are mutually independent such that $\Expt d(U_{i,O_i(j)})\leq D_{i,O_i(j)}$, let the genie-provided links support the rate matrix $\vec{R}_{\vec{O}}$ whose elements are
\begin{align}
\label{eq:ExcessRate}
R^*_{i,O_i(j)}=\left\{\begin{array}{cc}
C(D_{i,O_i(j)},U_{i,O_i(j)})& j\leq |\mathscr{Q}_i|,\,|\mathscr{Q}_i|>1\\
\Diamond&\mbox{otherwise}
\end{array}.
\right.
\end{align}
Then we have $\vec{D}\in \mathcal{D}^{**}_{mul}(\vec{O},\vec{R}^*)$.
\end{theorem}

\textit{Remark:} This theorem also implies $\bigcup_{\vec{O}}\mathcal{D}^{*}_{mul}(\vec{O})\subseteq
\mathcal{D}_{mul}\subseteq
\bigcup_{\vec{O}}\mathcal{D}^{**}_{mul}({\vec{O}},\vec{R}^*)$. It in fact provides more than one outer bound, one for each set of $V_{i,j}$ random 
variables, resulting in a rather powerful bounding tool. The auxiliary random variables $U_k$\rq{}s are used in constructing the channel code and the source code,  
and thus the genie-provided links are also parametrized by these random variables.

For certain distortion measures, significant simplifications can be made. 
The next result states that a separation-based scheme
is approximately optimal, universally across all distortion values,
for the quadratic distortion measure where the source alphabet and the reconstruction alphabet are reals\footnote{Our proofs for the JSCMUD and JSCMMD problems rely only on weak typicality instead of strong typicality, thus the result can be extended to the continuous sources and channels with continuous alphabets and unbounded distortion measures under the technical condition that for each source $S_i$, for all letters $\hat{s}_i\in\hat{\mathcal{S}}_i$, $\Expt d(S_i,\hat{s}_i)<\infty$.
This ``bounded expected distortion" condition  \cite{Wyner:78} assures that the asymptotically small decoding error probability does not cause significant change in the distortion behavior.}. Note that the sources need not
be Gaussian. 

\begin{theorem}
\label{corollary:Gaussian}
Let $\vec{D}$ be an achievable distortion matrix, for which $\vec{O}$ is the corresponding orders induced by $\vec{D}$. 
Let the sources $S_i$'s satisfy the condition that for all letters $\hat{s}_i\in\hat{\mathcal{S}}_i$, $\Expt (S_i-\hat{s}_i)^2<\infty$.
Let the genie-provided links support the rate matrix $\vec{R}_{\vec{O}}$ whose elements are
\begin{align}
R^*_{i,O_i(j)}=\left\{\begin{array}{cc}
1/2 \mbox{ bit}& j\leq |\mathscr{Q}_i|,\,|\mathscr{Q}_i|>1\\
\Diamond&\mbox{otherwise}
\end{array}.
\right.
\end{align}
We have $\vec{D}\in \mathcal{D}^{**}_{mul}(\vec{O},\vec{R}^*)$ under the mean squared error distortion measure.
\end{theorem}

In the simplest case where a single node broadcasts a Gaussian source
to a set of receivers, this result essentially reduces to Corollary 1
given in \cite{TianDiggaviShamai:09}. The
intuitive translation of the above result is that when a genie helps
the separation-based scheme by providing half a bit information for
each receiver, and at the same time, all the receivers with better quality
reconstructions receive this information for free, then the
genie-aided separation-based scheme is as good as the optimal ones.
For any fixed network, the approximation in Theorem \ref{corollary:Gaussian} holds regardless of the quality of the channel. As such, this result is more useful in the high resolution regime for large networks, when the genie-provided links become negligible compared to the original communication network. 

In the remainder of the section we focus on the proof of Theorem \ref{theorem:multicast}, since Theorem \ref{corollary:Gaussian} can be directly obtained by using Gaussian auxiliary random variables $V$'s in Theorem \ref{theorem:multicast}.

\begin{IEEEproof}[Proof of Theorem \ref{theorem:multicast}]
To simplify the notation, let us first consider a single source $S$; assume for the time-being that
the joint source-channel encoding procedure is still performed on other sources. 
Without loss of generality, assume the destination nodes of source $S$ are
$1,2,\ldots,K$; moreover, the distortions, which are achieved by this given source-channel joint 
code, are ordered as $D_{1}\geq D_{2}\geq \ldots\geq D_{K}$.

A set of auxiliary random variables are chosen in the 
alphabet $\mathcal{S}$ such that,
\begin{align}
\label{eq:AuxVarDef}
U_K=V_K,\quad U_k=U_{k+1}+V_k,\quad k=1,2,\ldots,K-1,
\end{align}
where $V_k$'s are random variables in the alphabet $\mathcal{S}$,
{\em independent of everything else}; furthermore, they have to satisfy
$\Expt d(U_k)\leq D_k$.

Consider a joint source-channel code which induces the distortion vector $(D_1,D_2,\ldots,D_K)$ for source $S$, whose reconstructions 
are $\hat{S}^m_{1},\hat{S}^m_{2},\ldots,\hat{S}^m_{K}$. The transition probability
$P(\hat{S}^m_{1},\hat{S}^m_{2},\ldots,\hat{S}^m_{K}|S^m)$ can be viewed as a broadcast channel, denoted as $P_{bc}$. 
We need the following lemma, whose proof will be given shortly. The asymptotically small quantities $\delta,\epsilon$ are omitted in the sequel, which are inconsequential.

\begin{lemma}
\label{lem:MulLemma2}
The following degraded message set broadcast rates can be (asymptotically) supported on $P_{bc}$
\begin{align}
&R^c_1=I(S^m+U^m_1;S^m)-mC(D_1,U_1)\nonumber\\ 
&R^c_k=I(S^m+U^m_k;S^m|S^m+U^m_{k-1})-mC(D_k,U_k),\nonumber\\
&\qquad\qquad\qquad\qquad\qquad k=2,3,\ldots,K. \label{eq:MulLemma2}
\end{align}
Moreover, these rates can be achieved by a random superposition code based on the joint distribution
$P(S^m+U^m_1,S^m+U^m_2,\ldots,S^m+U^m_{K-1},S^m)$.
\end{lemma}

Though this lemma is regarding the channel $P_{bc}$, in a manner similar to the proof for general network unicast,
we can conclude that on the original network, when all the other encoders still perform the original joint source-channel encoding, 
the communication channel from source $S$ to its destinations can support these rates per $m$ source samples. This is because the broadcast channel $P_{bc}$ is simply 
the original communication channel with certain additional operations on the block level. 
Thus together with the genie-provided links, we can send messages from source $S_i$ to its destinations at rates
\begin{align}
&R^{(m)}_1=I(S^m+U^m_1;S^m)\nonumber\\
&R^{(m)}_k=I(S^m+U^m_k;S^m|S^m+U^m_{k-1}),\, k=2,3,\ldots,K. 
\end{align}
The rates $(R^{(m)}_1,R^{(m)}_2,\ldots,R^{(m)}_K)$ are exactly the (asymptotic) source coding rates per $m$-samples in a successive refinement random code \cite{EquitzCover:91} constructed using the distribution $P(S^m+U^m_1,S^m+U^m_2,\ldots,S^m+U^m_K)$. Thus the distortion $\Expt d(S+U_k-S)=\Expt d(U_k)\leq D_k$ is achievable using the separation 
approach in this genie-aided network, if this successive refinement source code is used.

It remains to argue that if all the users simultaneously replace the original joint source-channel 
codes with the newly constructed channel codes, the rates that can be supported are
still the same. This is indeed true, because in Lemma \ref{lem:MulLemma2},
we only rely on the joint typicality on the block level when the channel input is of distribution $P(S^n)$. 
This however does not change if all the users replace the joint source-channel 
codes with their newly constructed channel codes, since these superposition channel codes preserve the joint typicality according to $P(S^m_1,S^m_2,\ldots,S^m_K)$. 
This completes the proof, except for Lemma \ref{lem:MulLemma2}. 
\end{IEEEproof}

To prove Lemma \ref{lem:MulLemma2}, we first give an auxiliary lemma.
\begin{lemma}
\label{lem:MulLemma1}
Let $S^m$, $U^m_i$ and $\hat{S}^m_i$ be specified as earlier, then we have for $i=1,2,\ldots,K$
\begin{align}
&I(S^m+U^m_1;S^m)-I(S^m+U^m_1;\hat{S}^m_i) &\nonumber\\
&\qquad\qquad\qquad\qquad\qquad\leq mC(D_i,U_1),\label{eq:MulLemma11}\\
&I(S^m+U^m_k;S^m|S^m+U^m_{k-1})\nonumber\\
&\qquad\qquad-I(S^m+U^m_k;\hat{S}^m_i|S^m+U^m_{k-1})\nonumber\\
&\qquad\qquad\qquad\leq  mC(D_i,U_k), \,\, k=2,\ldots,K-1,\label{eq:MulLemma12}\\ 
&I(S^m+U^m_K;S^m|S^m+U^m_{K-1})\nonumber\\
&\qquad\qquad-I(S^m;\hat{S}^m_i|S^m+U^m_{K-1}) \nonumber\\
&\qquad\qquad\qquad\qquad\qquad\leq mC(D_i,U_K).\label{eq:MulLemma13}
\end{align}
\end{lemma}

\begin{IEEEproof}[Proof of Lemma \ref{lem:MulLemma1}]

We can write
$I(S^m+U^m_1;\hat{S}^m_{i},S^m)$ in two ways
\begin{align}
&I(S^m+U^m_1;\hat{S}^m_{i},S^m)\nonumber\\
&=I(S^m+U^m_1;\hat{S}^m_{i})+I(S^m+U^m_1;S^m|\hat{S}^m_{i}),\\
&I(S_i^m+U^m_1;\hat{S}^m_{i},S^m)\nonumber\\
&=I(S^m+U^m_1;S^m)+I(S^m+U^m_1;\hat{S}^m_{i}|S^m)\nonumber\\
&=I(S^m+U^m_1;S^m),
\end{align}
where $I(S^m+U^m_1;\hat{S}^m_{i}|S^m)=I(U^m_1;\hat{S}^m_{i}|S^m)=0$,
because the construction of the auxiliary random variable $U_1$
ensures that $U^m_1$ is independent of $(\hat{S}^m_{i},S^m)$, as seen in
\eqref{eq:AuxVarDef}. Thus we have
\begin{align}
&I(S^m+U^m_1;S^m)-I(S^m+U^m_1;\hat{S}^m_{i})\nonumber\\
&=I(S^mU^m_1;S^m|\hat{S}^m_{i})\nonumber\\
&\stackrel{(a)}{=}H(S^m+U^m_1|\hat{S}^m_{i})-H(U^m_1)\nonumber\\
&\leq H(S^m-\hat{S}^m_{i}+U^m_1)-H(U^m_1)\nonumber\\
&\leq \sum_{j=1}^mH(S(j)-\hat{S}_i(j)+U_1(j))-H(U_1(j))\nonumber\\
&=\sum_{j=1}^m I(S(j)-\hat{S}_i(j);S(j)-\hat{S}_i(j)+U_1(j))\nonumber\\
&\leq mC(D_i,U_1),\label{eqn:firstgap}
\end{align}
where $(a)$ follows again since $U_1^m$ is independent of
$(\hat{S}^m_{i},S^m)$, and the last step follows the concavity of
$I(X;Y)$ as a function of the marginal distribution. This proves (\ref{eq:MulLemma11}).

Note further that for $k=2,3,\ldots,K$, we have
\begin{align*}
&I(S^m+U^m_k;S^m,\hat{S}^m_i|S^m+U^m_{k-1})\nonumber\\
&=I(S^m+U^m_k;S^m|S^m+U^m_{k-1})\nonumber\\
&\qquad\qquad+I(S^m+U^m_k;\hat{S}^m_i|S^m,S^m+U^m_{k-1})\nonumber\\ 
&=I(S^m+U^m_k;S^m|S^m+U^m_{k-1}),
\end{align*}
as well as
\begin{align*}
&I(S^m+U^m_k;S^m,\hat{S}^m_i|S^m+U^m_{k-1})\nonumber\\
&= I(S^m+U^m_k;\hat{S}^m_i|S^m+U^m_{k-1})\nonumber\\
&\qquad\qquad+I(S^m+U^m_k;S^m|\hat{S}^m_i,S^m+U^m_{k-1}).
\end{align*}
It follows that
\begin{align}
&I(S^m+U^m_k;S^m|S^m+U^m_{k-1})\nonumber\\
&\qquad-I(S^m+U^m_k;\hat{S}^m_i|S^m+U^m_{k-1})\nonumber\\
&=I(S^m+U^m_k;S^m|\hat{S}^m_i,S^m+U^m_{k-1}).
\end{align}
Thus we have
\begin{align}
&I(S^m+U^m_k;S^m|\hat{S}^m_i,S^m+U^m_{k-1})\nonumber\\
&\stackrel{(b)}{=}H(S^m|\hat{S}^m_i,S^m+U^m_{k-1})-H(S^m|\hat{S}^m_i,S^m+U^m_k)\nonumber\\
&\leq H(S^m|\hat{S}^m_i)-H(S^m|\hat{S}^m_i,S^m+U^m_k)\nonumber\\
&=I(S^m;S^m+U^m_k|\hat{S}^m_i)\nonumber\\
&=H(S^m+U^m_k|\hat{S}^m_i)-H(U^m_k)\nonumber\\
&\leq H(S^m-\hat{S}^m_i+U^m_k)-H(U^m_k)\nonumber\\
&\leq mC(D_i,U_k),\label{eqn:secondgap}
\end{align}
where $(b)$ is due to the Markov string $S^m+U_{k-1}^m \leftrightarrow
S^m+U_k^m \leftrightarrow S^m \leftrightarrow \hat{S}^m_i$. This proves (\ref{eq:MulLemma12}).

Because of the Markov string $\hat{S}^m_i \leftrightarrow S^m \leftrightarrow S^m+U^m_K\leftrightarrow S^m+U^m_{K-1}$, we have 
\begin{align}
&I(S^m;\hat{S}^m_i|S^m+U^m_{K-1})-I(S^m+U^m_K;\hat{S}^m_i|S^m+U^m_{K-1})\nonumber\\
&=H(\hat{S}^m_i|S^m+U^m_{K-1},S^m+U^m_K)\nonumber\\
&\qquad\qquad\qquad-H(\hat{S}^m_i|S^m+U^m_{K-1},S^m)\nonumber\\
&=H(\hat{S}^m_i|S^m+U^m_K)-H(\hat{S}^m_i|S^m)\nonumber\\
&=H(\hat{S}^m_i|S^m+U^m_K)-H(\hat{S}^m_i|S^m,S^m+U^m_K)\nonumber\\
&=I(\hat{S}^m_i;S^m|S^m+U^m_K)\geq 0,
\end{align}
and it follows that
\begin{align}
&I(S^m+U^m_K;S^m|S^m+U^m_{K-1})-I(S^m;\hat{S}^m_i|S^m+U^m_{K-1})\nonumber\\
&\leq I(S^m+U^m_K;S^m|S^m+U^m_{K-1})\nonumber\\
&\qquad\qquad\qquad-I(S^m+U^m_K;\hat{S}^m_i|S^m+U^m_{K-1})\nonumber\\
&\leq mC(D_i,U_K). \label{eqn:thirdgap}
\end{align}
This proves (\ref{eq:MulLemma13}).
\end{IEEEproof}

\begin{IEEEproof}[Proof of Lemma \ref{lem:MulLemma2}]

We shall use the distribution $P(S^m+U^m_1,S^m+U^m_2,\ldots,S^m+U^m_{K-1},S^m)$ to construct 
superposition broadcast channel code on the broadcast channel $P_{bc}$ for a degraded message set. 
The rates (per length-$m$ block) for these messages within the degraded message set are (asymptotically)
\begin{align}
&R^c_1=I(S^m+U^m_1;S^m)-mC(D_1,U_1),\label{eqn:R1}\\
&R^c_i=I(S^m+U^m_i;S^m|S^m+U^m_{i-1})-mC(D_i,U_i),\nonumber\\
&\qquad\qquad\qquad\qquad\qquad i=2,3,\ldots,K,\label{eqn:RK}
\end{align}
which need to be shown to be indeed achievable on $P_{bc}$. 

Since this channel itself is not degraded, we have to show that the superposition coding
scheme succeeds for all the receivers. To see this, observe that for the $i$-th
receiver, we have
\begin{align}
I(S^m+U^m_1;S^m)-I(S^m+U^m_1;\hat{S}^m_i)\leq mC(D_i,U_1),
\end{align}
by Lemma \ref{lem:MulLemma1}. It follows that
\begin{align}
&I(S^m+U^m_1;\hat{S}^m_i)-R^c_1\nonumber\\
&=I(S^m+U^m_1;\hat{S}^m_i)-I(S^m+U^m_1;S^m)+mC(D_1,U_1)\nonumber\\
&\geq mC(D_1,U_1)-mC(D_i,U_1)\geq 0,
\end{align}
where the last inequality is straightforward by noticing
\begin{align}
C(D,N)\geq C(D',N),
\end{align}
when $D\geq D'$. Thus the $i$-th receiver,  $i\geq 1$,  can indeed decode the first message.

Similarly, we have for $i\geq k$
\begin{align}
&I(S^m+U^m_k;\hat{S}^m_i|S^m+U^m_{k-1})-R^c_k\nonumber\\
&=I(S^m+U^m_k;\hat{S}^m_i|S^m+U^m_{k-1})\nonumber\\
&\qquad-I(S^m+U^m_k;S^m|S^m+U^m_{k-1})+mC(D_k,U_k)\nonumber\\
&\geq mC(D_k,U_k)-mC(D_i,U_k)\geq 0,
\end{align}
and thus we conclude the $i$-th receiver can decode the messages
$1,2,\ldots,i$. The $K$-th receiver does not pose any additional
difficulty. Thus indeed the rates specified in
(\ref{eqn:R1})-(\ref{eqn:RK}) can be supported on $P_{bc}$, and the proof is complete.
\end{IEEEproof}

\section{Concluding Remarks}
\label{sec:conclusion}

We considered the optimality of the source-channel separation architecture
in networks, and showed that the separation approach is optimal for the problems of
distributed network joint source-channel coding and joint source-channel multiple 
unicast with distortions. Moreover, the separation approach is also
approximately optimal for the problem of joint source-channel multiple multicast with distortions under
certain distortion measures. The results in this work are obtained without explicit characterizations of the underlying
regions. The source coding problem extracted from the distributed network
source coding scenario implies that the interactive coding aspect
needs to be carefully incorporated, which suggests a distinct line of research direction into network
source coding.

For notational and conceptual simplicity, we made many assumptions which are not strictly necessary. 
We believe the results can be extended to more general cases with some minimal efforts.
\begin{itemize}
\item \textbf{Distributed network joint source-channel coding:} The synchronization requirement among sources can be removed, {\em i.e.}, the source bandwidths do not have to be the same throughout the whole network. The reconstructions of a source $S_i$ can be under different distortion measures; in fact the distortion measures can be defined on multiple sources, such as to reconstruct $(S_1-S_2)$. The restriction on the sources and the channels being finite-alphabet may be relaxed using the techniques in \cite{ElGamalKimBook}.
\item \textbf{Joint source-channel multiple unicast with distortions:} The synchronization requirement among sources and channels can be removed and the memoryless requirement on the channel can be relaxed to channels with finite memory (see \cite{TianArxiv:10} for an outline). As mentioned, the restriction on the finite alphabets can be relaxed. The condition that each source is to be reconstructed at one destination can be relaxed to some extent: when each source is to be reconstructed at multiple destinations but at the exact same distortion, then the source-channel separation architecture is still optimal. 
\item \textbf{Joint source-channel multiple multicast with distortions:} Similar to the JSCMUD case, the synchronization, the memoryless channel, and the finite-alphabet requirement can be relaxed. The condition that each source is to be reconstructed under the same distortion measure can be relaxed  to different distortion measures. 
If some of the reconstructions of a source $S_i$ are specified to have the same distortion a priori, then the approximation upper bound can be improved.
\end{itemize}

In the point-to-point setting, the source and the channel are 
specified by their statistical behaviors alone; however in the network setting, 
the new components of the connectivity structure among nodes 
and the source-demand coding requirements are introduced.
Our result in DNJSCC treats the source statistics and these network components 
as a whole, and the channel statistics as the other, resulting in the separation between
 a complex network source coding problem and multiple conventional point-to-point channel coding problems. 
In contrast, the result in JSCMUD treats the channel statistics and the network components as a whole, and
the source statistics as the other, resulting in the separation between a complex
network channel coding problem and multiple conventional point-to-point source coding problems. 
These separations are not the only possibilities, and one can choose to separate in a different manner. 
In this work we have not considered transmitting generally correlated sources over a general channel network, and it is unclear whether
there exist scenarios for which a separation architecture is optimal or approximately optimal. Thus the problem of source-channel separation is by no means solved, and 
it in fact calls for further investigation.

\section*{Acknowledgment}

The authors wish to thank Associate Editor Young-Han Kim and the anonymous reviewers for their constructive comments.

\bibliographystyle{IEEEbib}

\begin{IEEEbiographynophoto}{Chao Tian}(S'00, M'05, SM'12) received the B.E. degree in Electronic Engineering from Tsinghua University, Beijing, China, in 2000 and the M.S. and Ph. D. degrees in Electrical and Computer Engineering from Cornell University, Ithaca, NY in 2003 and 2005, respectively. 

Dr. Tian was a postdoctoral researcher at Ecole Polytechnique Federale de Lausanne (EPFL) from 2005 to 2007. He joined AT\&T Labs--Research, Florham Park, New Jersey in 2007, where he is now a Senior Member of Technical Staff. He is currently also an Adjunct Associate Professor at Columbia University and an Associate Editor for {\sc the IEEE Signal Processing Letters}. His research interests include multi-user information theory, joint source-channel coding, quantization design and analysis, as well as image/video coding and processing. 

Dr. Tian received the Liu Memorial Award at Cornell University in 2004, and the AT\&T Key Contributor Award in 2010 and 2011. 
\end{IEEEbiographynophoto}

\begin{IEEEbiographynophoto}{Jun Chen}(S'03, M'06) received the B.E. degree with honors in communication engineering from Shanghai Jiao Tong University, Shanghai, China, in 2001 and the M.S. and Ph.D. degrees in electrical and computer engineering from Cornell University, Ithaca, NY, in 2004 and 2006, respectively.

He was a Postdoctoral Research Associate in the Coordinated Science Laboratory at the University of Illinois at Urbana-Champaign, Urbana, IL, from 2005 to 2006, and a Postdoctoral Fellow at the IBM Thomas J. Watson Research Center, Yorktown Heights, NY, from 2006 to 2007. He is currently an Associate Professor of Electrical and Computer Engineering at McMaster University, Hamilton, ON, Canada. His research interests include information theory, wireless communications, and signal processing.

He received several awards for his research, including the Josef Raviv Memorial Postdoctoral Fellowship in 2006, the Early Researcher Award from the Province of Ontario in 2010, and the IBM Faculty Award in 2010.
\end{IEEEbiographynophoto}

\begin{IEEEbiographynophoto}{Suhas N. Diggavi}
(S'93, M'99, F'13) received the B. Tech. degree in electrical
engineering from the Indian Institute of Technology, Delhi, India, and
the Ph.D. degree in electrical engineering from Stanford University,
Stanford, CA, in 1998.

After completing his Ph.D., he was a Principal Member Technical Staff 
in the Information Sciences Center, AT\&T Shannon Laboratories, Florham
Park, NJ. Since then he had been in the faculty of the School of
Computer and Communication Sciences, EPFL, where he directed the
Laboratory for Information and Communication Systems (LICOS). He is
currently a Professor, in the Department of Electrical Engineering, at
the University of California, Los Angeles. His research interests
include wireless communications networks, information theory, network
data compression and network algorithms.

He is a co-recipient of the 2013 IEEE Information Theory Society \& Communications Society Joint Paper Award, the 2013 ACM International Symposium on Mobile Ad Hoc Networking and Computing (MobiHoc) best paper award, the 2006 IEEE Donald Fink prize paper award, 2005 IEEE Vehicular Technology Conference best paper award, the Okawa foundation research award and is a Fellow of the IEEE. He has served on the editorial board for {\sc the IEEE Transactions on Information Theory}, {\sc the ACM/IEEE Transactions on Networking} and {\sc the IEEE Communication Letters}, a guest editor for {\sc the IEEE Journal on Selected Topics in Signal Processing} and was the Technical Program Co-Chair for 2012 IEEE Information Theory Workshop. He has 8 issued patents.

\end{IEEEbiographynophoto}

\begin{IEEEbiographynophoto}{Shlomo Shamai (Shitz)}(S'80, M'82, SM'89, F'94) 
received the B.Sc., M.Sc., and Ph.D. degrees in electrical engineering from the Technion---Israel Institute of Technology, in 1975, 1981 and 1986 respectively.

During 1975-1985 he was with the Communications Research Labs, in the capacity of a Senior Research Engineer. Since 1986 he is with the Department of Electrical Engineering, Technion---Israel Institute of Technology, where he is now a Technion Distinguished Professor, and holds the William Fondiller Chair of Telecommunications. His research interests encompasses a wide spectrum of topics in information theory and statistical communications.

Dr. Shamai (Shitz) is an IEEE Fellow, a member of the Israeli Academy of Sciences and Humanities and a Foreign Associate of the US National Academy of Engineering. He is the recipient of the 2011 Claude E. Shannon Award and the 2014 Rothschild Prize in Mathematics/Computer Sciences and Engineering.
He has been awarded the 1999 van der Pol Gold Medal of the Union Radio Scientifique Internationale (URSI), and is a co-recipient of the 2000 IEEE Donald G. Fink Prize Paper Award, the 2003, and the 2004 joint IT/COM societies paper award, the 2007 IEEE Information Theory Society Paper Award, the 2009 European Commission FP7, Network of Excellence in Wireless COMmunications (NEWCOM++) Best Paper Award, and the 2010 Thomson Reuters Award for International Excellence in Scientific Research.  He is also the recipient of
1985 Alon Grant for distinguished young scientists and the 2000 Technion Henry Taub Prize for Excellence in Research.
He has served as Associate Editor for {\sc the Shannon Theory of the IEEE Transactions on Information Theory}, and has also served twice on the Board of Governors of the Information Theory Society.
He is a member of the Executive Editorial Board of the IEEE Transactions on Information Theory

\end{IEEEbiographynophoto}

\end{document}